\documentclass{article}
\usepackage[
    colorlinks,%
    linkcolor=blue,citecolor=red,urlcolor=blue,
    hyperindex,%
    plainpages=false,%
    bookmarksopen,%
    bookmarksnumbered,%
]{hyperref}

\usepackage{authblk}

\usepackage{color}
\usepackage{comment}

\usepackage[utf8]{inputenc}

\usepackage{graphicx}
\graphicspath{{.}}
\usepackage{amsmath,amssymb}

\usepackage{footnpag} 
\usepackage[all]{xy}  
\def\xybiglabels{\def\labelstyle{\textstyle}} 

\xyoption{2cell}
\UseAllTwocells

\usepackage{epsfig}
\usepackage{rlepsf}

\newtheorem{theorem}{Theorem}
\numberwithin{theorem}{subsection}
\newtheorem{conjecture}[theorem]{Conjecture}

\newtheorem{lemma}[theorem]{Lemma}
\newtheorem{varexample}[theorem]{Example}

\newtheorem{definition}[theorem]{Definition}
\newtheorem{varremark}[theorem]{Remark}
\newenvironment{remark}{\begin{varremark}\em}{\em\end{varremark}}
\newenvironment{proof}{\textbf{Proof}:}{$\Box$\vskip 10pt}

\newcommand{\ra}{\rightarrow}

\newcommand{\opname}[1]{\operatorname{#1}}
\newcommand{\cat}[1]{\boldsymbol{\opname{#1}}}
\newcommand{\squaremor}[1]{\begin{array}{|lcr|}\hline & {#1} & \\
    \hline\end{array}}

\newcommand{\wquot}{/\!\!/}

\newcommand{\G}{\mathcal{G}}

\newcommand{\C}{\cat{C}}

\newcommand{\act}{\blacktriangleright}

\newcommand{\R}{\mathbb{R}}

\begin{document}

\title{2-Group Actions and Moduli Spaces of Higher Gauge Theory}
\author{\sc Jeffrey C.\ Morton}

\affil{{\it Mathematics Department} \\
       {\it Buffalo State University}\\
        {\it 1300 Elmwood Avenue}\\
        {\it Buffalo NY 14222 USA}\\mortonjc@buffalostate.edu }

\author{\sc Roger Picken}

\affil{{\it Center for Mathematical Analysis, Geometry and Dynamical Systems}\\
{\it Mathematics Department}\\{\it Instituto Superior T\'{e}cnico, Universidade de Lisboa}\\
{\it Avenida Rovisco Pais, 1049-001 Lisboa, Portugal} \\
roger.picken@tecnico.ulisboa.pt}

\maketitle

\begin{abstract}
A framework for higher gauge theory 
based on a 2-group is presented, by constructing 
a groupoid of connections on a manifold  acted on by a 2-group of gauge 
transformations, following 
previous work by the authors where the general notion of the action of 
a 2-group on a category was defined. The connections are discretized, given by 
assignments of 2-group data to 1- and 2-cells coming from a given cell structure 
on the manifold, and likewise the gauge transformations are given by 2-group
assignments to 0-cells. The 2-cells of the manifold are endowed with a bigon 
structure, matching 
the 2-dimensional algebra of squares which is used for calculating with 2-group data.
Showing that the action of the 2-group of gauge transformations on the groupoid of
connections  is well-defined is the central result. The effect, on the 
groupoid of connections, of 
changing the discretization is studied, and partial results and conjectures 
are presented around this issue. The transformation double category that 
arises from the action of a 2-group on a category, 
as defined in previous work by the authors, is described for the case at hand, where
it becomes a transtormation double groupoid. 
Finally, examples of the construction are given for 
simple choices of manifold:
the circle, the 2-sphere and the torus.

\end{abstract}

\newpage

\tableofcontents

\section{Introduction}

\subsection{General Background}

The notion of symmetry plays a fundamental role throughout mathematics
and its applications. In this paper, we investigate its role in 
{\em higher gauge theory} (HGT), a generalization of ordinary gauge theory
of considerable recent interest in physics. For an introduction to HGT, see
\cite{baez-huerta}.

Gauge theory is the study of certain geometric structures on
manifolds, namely connections. If $G$ is a group and $M$ a connected
manifold, the \textit{moduli space} of principal $G$-bundles over $M$
equipped with a connection is a space in which each point represents a
choice of principal $G$-bundle with connection. These have a wide
variety of applications. Physically, gauge theories are used to
represent particles and fields on a background spacetime $M$, so that
a connection on a principal bundle represents a (classical) state of
such a physical system. In geometric topology, one can study a
manifold by means of such geometric structures and their moduli
spaces.

For such a theory on a given manifold, the symmetries of the moduli
space are very important. In a physical
situation modeled by a gauge theory, there may be many distinct
connections on $G$-bundles which represent physically
indistinguishable states. This situation can be described as a
groupoid $\mathcal{A}(M,G)$. The objects of this groupoid are
principal $G$-bundles over $M$ equipped with a connection.  Other
groupoids, such as $\mathcal{A}_0(M,G)$, whose objects are principal
$G$-bundles with flat connection, also play important roles. In each
case, the morphisms are symmetries of such configurations, which are
the gauge transformations. One way of looking at this groupoid is in
terms of global symmetry: each gauge transformation is a
transformation of the bundle which fixes the base space, and 
transforms the connection in a compatible way.
 These form a group which acts on the space of
connections. The groupoid is, from this point of view, the
\textit{transformation groupoid} (also called the \textit{action
  groupoid}) associated to this action.

More recently, it has become common, instead of a moduli space, to
refer to a \textit{moduli stack}. While we do not wish to enter into
the general theory of stacks here, we merely note that they are
associated to equivalence classes of groupoids, so that the groupoids
described above may be seen  to represent the same moduli stack. 
 This
bypasses the problems associated with both the approach of using the
\textit{fine moduli space}, in which all connections appear as
distinct points, and the approach using the
\textit{coarse moduli space}, the quotient in
which isomorphic connections are identified. The problem with the fine
moduli space is precisely that such isomorphisms are not apparent,
while the coarse moduli space may have inconvenient geometric features
(as quotients of manifolds by group actions need not themselves be
manifolds). However, if one takes these two spaces as the 
set of objects of
two different groupoids (the second of which will contain only
automorphisms of its objects), these will be  
equivalent as categories,
and will represent the same moduli stack, which is one motivation for
this approach.

Now, in higher gauge theory, the role played by the structure group $G$ in
ordinary gauge theory is taken instead by a higher-categorical object,
$\G$, called a 2-group (or, more generally, an $n$-group).  Yetter
\cite{yettqft}, and Martins and Porter \cite{martinsporter} have
discussed an extension, or ``categorification'' of the
Dijkgraaf-Witten theory for $G$-bundles, which gives invariants of the type
described above, replacing $G$ by a 2-group $\G$. We will recall the definition
and basic facts about 2-groups in Section
\ref{sec:crossedmodules}. Here we only remark that there is a general
program of reproducing the machinery of gauge theory for these
higher-categorical groups. For a discussion of such invariants based
on higher gauge theory, see for instance work by Baez and
Schreiber \cite{basch-hgt}, and Bartels \cite{bartels}.

To understand the symmetries of such a theory, we need
to generalize the notion of symmetry which is relevant in the context
of groups (and group actions in particular) to the context of
2-groups. This paper is the second in a series of three closely
connected works which approach this question.

In the first paper of the series, we described the notion of a 2-group
action on a category, and the categorified analog of the
transformation groupoid associated to a group action
\cite{morton-picken-i}. This established that there is a double
category associated to every such action, which we call the
transformation double category.

Our aim in this paper is to describe the analog, in a higher gauge
theory based on a 2-group $\G$, of the symmetry group of all gauge
transformations which acts on the moduli space of principal
$G$-bundles with connection on a manifold $M$, in ordinary gauge theory.
The analog will take the form of a 2-group $\cat{Gauge}$ acting
on a category $\cat{Conn}$ of connections, which we may call the
``moduli 2-space''. 

In the third paper of the series, we will show a key result about the
transformation double category $\cat{Conn}\wquot \cat{Gauge}$
associated to this 2-group action. In particular, we demonstrate that
it is equivalent to a double category which emerges naturally as a
result of treating connections in terms of transport
2-functors. Indeed, by making suitable choices, we can arrange that
the equivalence is a strict isomorphism. This result is an analog of a
well-known fact in ordinary gauge theory showing the equivalence of
two ways of describing the groupoid of connections
\cite{schreiberwaldorfii}. 

To make sense of all this, we will need to establish the nature of what
we have just referred to as the ``moduli 2-space'' on which a symmetry
2-group may act. Just as it makes sense to consider group actions on
an object of any category, the natural context for 2-group actions is
on objects of 2-categories. However, the familiar context for group
actions is the category $\cat{Set}$ of sets and functions - which can
be refined in case the group and the set have additional structure,
such as a Lie group acting on a manifold. Similarly, the natural
situation for a 2-group action is in the 2-category $\cat{Cat}$, of
categories, functors, and natural transformations. This, again, may be
refined in the case that these have extra structure, geometric or
otherwise.

There are some new features that appear in moving from actions of
groups to actions of
2-groups. As we will see, there are two intrinsically rather
different kinds of symmetry relation between $\G$-connections, which
in the literature on higher gauge theory are often both described as
``gauge transformations''.  A distinction occurs here which cannot
arise for ordinary group actions on sets and set-based structures, but
will naturally occur in higher gauge theory.  Namely, certain of the
symmetries between connections occur naturally as morphisms in the
2-space (i.e. category) of connections itself, while others arise from
the action. A similar distinction will occur in a higher gauge theory
based on an $n$-group for any $n \geq 2$.

That is, two connections may be related, as objects of $\cat{Conn}$,
by a morphism in $\cat{Conn}$ itself. Or, on the other hand, they may
be related by the action of an object in the 2-group $\cat{Gauge}$. In
the construction we give here of the transformation double category of
a 2-group action, these will be the horizontal and vertical morphisms
respectively.  We give advance notice here that we also refer to these
occasionally as, respectively, \textit{costrict} and \textit{strict}
gauge transformations.  What are commonly called gauge transformations
are in general mixtures of these two types.

Another layer of structure exists in this setting, also commonly
included under the general umbrella term ``gauge transformation'',
which is represented by the morphisms of the symmetry 2-group
$\cat{Gauge}$. These we refer to simply as squares, or
occasionally as \textit{gauge modifications}. They may be thought of
either as relating morphisms of $\cat{Conn}$ (that is, costrict gauge
transformations), or as relating objects of $\cat{Gauge}$ (that is,
strict gauge transformations), depending on the direction in the
double category in which we take the source and target of a square.

The terminology of strict and costrict gauge transformations, and of
gauge modifications, is chosen by analogy with that used in the
context of 2-functors between 2-categories: a natural transformation
is a map between 2-functors, and a modification is a map between
natural transformations. In a forthcoming paper
\cite{morton-picken-iii}, we will show that this parallel is more than
simply an analogy. There we will show that the transformation
structure asssociated to the 2-group action constructed here is
equivalent to a certain double category of transport functors, strict and
costrict natural transformations, and modifications.

\subsection{Overview}

In section \ref{sec:crossedmodules}
 we recall the notion of a 2-group $\G$ in its guise as a crossed module,
and then introduce a convenient 2D calculus based on squares labelled with
crossed module data. In this context we note that in higher categories one
is often faced with different choices for the shapes of higher morphisms, e.g. for 
2-morphisms one can consider a simplicial shape (triangles), a cubical shape (squares),
or a globular shape (bigons). The special squares of the 2D calculus combine features
of squares and bigons, as we will see.

In section \ref{sec:catconn}, we embark on our study of HGT based on a 2-group $\G$ and with 
discretized connections. First, in subsection \ref{sec:Hdisc-conn-transp}, 
we introduce the cell structure on
the manifold $M$, which forms the basis for the discretization. The main feature 
of this cell structure is what we call a bigon structure on the 2-cells, which is set up
to match the  2D algebra of squares. Discretized connections on $M$ are then 
defined to be suitable assignments of 2-group data to the 1- and 2- cells of $M$.
The aim is to avoid the analytic issues associated with infinite-dimensional spaces
of connections in general, which need to be addressed, for example,  by the
use of Frechet manifolds or diffeological spaces.

We construct a category of connections $\cat{Conn}$ with discretized connections 
as its objects. The morphisms of this category can be thought of as certain types of 
gauge transformations, and this entire category, including the morphisms,
 constitutes the structure which
corresponds to the \textit{fine} moduli space. Thus we might describe $\cat{Conn}$ as
a moduli ``2-space'' to suggest this correspondence, with the
understanding that a 2-space is merely a category in which both
objects and morphisms form spaces (e.g. topological spaces, manifolds, schemes, etc. 
depending on the context).
We also show that the category $\cat{Conn}$ has invertible morphisms,
i.e. it is a groupoid. 

Then in subsection \ref{sec:2group-action}, the 2-group of gauge transformations
$\cat{Gauge}$ is introduced, given by assignments of 2-group data to 0-cells of $M$. 
We recall the general notion of the action of  a 2-group on a category 
from our previous work \cite{morton-picken-i}, 
and define the action of $\cat{Gauge}$ on $\cat{Conn}$. A main
result is proving that this is indeed a valid action.

In subsection \ref{sec:ch-discr}, we address the issue of what happens when the choice of 
discretization of $M$ is changed. We provide partial answers for how this affects 
the category $\cat{Conn}$, along with some conjectures about more general statements
and the continuum limit.

In section \ref{sec:transdoublegpd} we recall our construction \cite{morton-picken-i} 
of a transformation double category from any
action of a 2-group $\G$ on a category $\cat{C}$. 
When the category $\cat{C}$ is a groupoid, the transformation double category becomes 
a transformation double groupoid.
This is the categorification of the 
transformation groupoid that arises from the action of a group $G$ on a set.
We then particularize to the case at hand and give details of the transformation
double groupoid for the action of $\cat{Gauge}$ on $\cat{Conn}$.
This explicit description is 
useful for two reasons. First, the double groupoid contains all
the information in the action, packaged in a convenient way, and so is
of intrinsic interest to understanding the symmetry of the 2-space
$\cat{Conn}$. Second, it will be the main object of study in our
forthcoming third paper in this series \cite{morton-picken-iii}, 
where we show it is equivalent
to a double groupoid constructed using an approach based on transport
2-functors.

In section \ref{sec:geom-examples}, we describe features of the
action of $\cat{Gauge}$ on $\cat{Conn}$ for a number of
elementary examples, when $M$ is  
the circle $S^1$, the sphere $S^2$, and the torus
$T^2$. These are chosen because they are manifolds where only the
first fundamental groupoid is nontrivial, where only the second
fundamental groupoid is nontrivial, and where both are
nontrivial. The examples are sufficient to illustrate the effect of each of
these elements of the homotopy type of a manifold on the resulting
construction. Moreover, they have special features that make them
particularly interesting in their own right.

The example of the circle
initially seems of limited importance,  since the most salient feature of 
2-group gauge theory is that it is possible to define
2-dimensional parallel transport over surfaces - hence, on manifolds
of dimension at least 2.  However, just as with 1-groups, the
$\G$-connections on the circle, and the various levels of gauge transformations 
between them, model the
adjoint action of the 2-group on itself. Indeed, the structure of the
circle as an oriented cell complex encodes precisely what we mean by
this adjoint action by conjugation, since the edge representing the
circle has its endpoints attached to each other with opposite
orientations on the two endpoints.

As a final remark, we note that all calculations are carried out using 
ordinary algebra and 
the 2D algebra of squares introduced in section \ref{sec:crossedmodules}. However, 
along the way we also
give various
pointers to some underlying 3D and 4D algebraic structures, not pursued in depth 
here.
See in particular Rem. \ref{rem:bigon-cyl}, Rem. \ref{rem:3Daction},
 Fig. \ref{fig:3D-circle} and 
Fig. \ref{fig:2gauge-torus}.

\subsection{Related Work}

There are a number of closely related articles that we wish to draw attention to. 

First, the use of double groupoids for approaching higher gauge theory is very 
much to the fore in an article by Soncini and Zucchini \cite{soncinizucchini}, although
the perspective there is rather different to ours. These authors describe HGT transports
on a manifold $M$ endowed with connection 1- and 2-forms,
in terms of double functors from a double groupoid of rectangles in $\R^2$ to
a double groupoid constructed from the chosen Lie crossed module. Gauge transformations are 
then given by double natural transformations and double modifications between these
double functors. 

A careful study of HGT along similar lines to ours was carried out by 
Bullivant, Calcada, Kádár, Faria Martins and Martin \cite{bullivant-et-al} in the 
context of their investigation
 of topological phases of matter in 3+1 dimensions. They 
also discretize the higher connections using  manifolds with an adapted cell structure, 
termed a 2-lattice structure. The 2-groups that they consider are finite, but this 
doesn't prevent a comparison with our approach without this restriction. 
They focus on transports along 2-disks and holonomies along 2-spheres, whereas our 
examples include also the circle and the torus. 
The gauge transformations in \cite{bullivant-et-al} correspond
to the morphisms in our category of connections and the objects of our 2-group of
gauge transformations, as will be described in section \ref{sec:catconn} below.
For their purpose they do not need the higher level of gauge transformations which 
in our approach are given by the morphisms of the gauge 2-group, and this is the 
most significant difference between the two approaches. 
However there are 
 also many similarities, and we will return to more detailed 
comparisons at appropriate points in the main text. 

We would also like to mention two articles by one of us with D. Bragança 
\cite{braganca-picken-i, braganca-picken-ii}, which concern the case where 
$M$ is a surface, possibly with boundary, and the group \cite{braganca-picken-i}
or 2-group \cite{braganca-picken-ii} is finite. There is an underlying
cell structure on $M$ captured by the notion of ``cut cellular surface'', and
the counting invariants (which tie in with the Yetter invariant \cite{yettqft})
and TQFT's that are the main focus in these articles
call out for an interpretation as the counting measure of the moduli 
spaces of  flat connections or flat 
higher  connections that are the subject of the present paper.

Finally, in work by one of us together with J. Nelson, see \cite{nelsonpicken} 
and references therein, models of quantum gravity in 2+1 dimensions are studied using traces
of holonomies (Wilson loops). These exhibit area phases relating loops that are  
homotopic on the spatial surface, which strongly suggests an 
underlying HGT mechanism.

\subsection{Considerations for Future Work}

It is our hope that examining this 2-group symmetry action in 
higher gauge theory will be
illuminating in its own right, and will serve as a practical illustration of our
earlier work \cite{morton-picken-i} on 2-group symmetry and transformation double
categories. More specifically, however, in future work \cite{morton-picken-iii}, 
we will show
the equivalence result referred to above, between the double groupoid
$\cat{Conn}\wquot \cat{Gauge}$ and a transport double groupoid.
 This will reveal that a slightly
unusual approach to categories of transport functors will be the
relevant one for higher gauge theory. Beyond this, we expect that our
double-category approach will be useful in extending 
other gauge theory constructions
 to higher gauge theories. In particular, the
construction of extended topological field theories using
2-linearization \cite{morton-2lin}, in contexts in which cobordism
categories become double categories of cobordisms with corners, as
previously studied by one of the authors \cite{morton-dlbicat}, appears
promising.
Another natural direction in which to extend our results is to include 
manifolds with boundary, like the surfaces with boundary of 
\cite{braganca-picken-i, braganca-picken-ii}, and to study examples 
with $M$ of dimension 3 or more, and non-vanishing curvature on 3-cells.

\section{Preliminaries on 2-Groups and Crossed Modules}
\label{sec:crossedmodules}

There are several different manifestations of 2-groups, namely 
as a certain type of 2-categories, 
as categorical groups (a certain type of category), or as crossed modules 
(an algebraic definition with no explicit categorical content). 
See \cite{morton-picken-i} for a detailed discussion of the relation between 
these different viewpoints. In this article we will mainly adopt the crossed module 
perspective and use a convenient 2D algebra based on squares labelled with
crossed module data.

\subsection{Crossed Modules and Calculus with Squares}

\begin{definition}
  A crossed module $\G$ consists of $(G,H,\rhd,\partial)$, where $G$ and
  $H$ are groups, $\rhd$ is an action of $G$ on $H$ by automorphisms
  and $\partial : H \ra G$ is a
  homomorphism, satisfying the following two conditions:
  \begin{eqnarray}
    \partial(g \rhd \eta) & =&  g \partial(\eta) g^{-1} \label{cm1}\\    
    \partial(\eta) \rhd \zeta & = &\eta \zeta \eta^{-1} \label{cm2}
  \end{eqnarray}
\end{definition}

Nontrivial examples of crossed modules arise, e.g., from central extensions of groups. See \cite[Examples 1.5-1.12]{martinspicken} for more examples.

Given a crossed module $\G$ we will be performing calculations using squares of the form
\begin{equation}
\xybiglabels \vcenter{\xymatrix@M=0pt@=3pc{\ar@{-} [d] _{} \ar@{-} [r]^{g} \ar@{} [dr]|\eta & \ar@{-} [d]^{} \\
\ar@{-} [r]_{g'}  & }}
\label{eq:square}
\end{equation}
where $g,g'\in G,\, \eta\in H$ and $\partial(\eta)= g'g^{-1}$. These squares are special cases of the squares of the double groupoid ${\cal D}(\G)$ of $\G$ \cite{brown-higgins-sivera,martinspickenii, morton-picken-i}, having the side edges of the squares labelled by $1_G\in G$ (displayed by their being unlabelled), instead of a generic $G$ element. Likewise omitting the label in the centre of the square denotes that it is labelled by $1_H$.

Horizontal and vertical composition of squares are given by:
\begin{equation*}
\xybiglabels \vcenter{\xymatrix @=3pc @W=0pc @M=0pc { \ar@{-}[r] ^{g_1} \ar@{-}[d]
_{} \ar@{}[dr]|{\eta_1} & \ar@{-}[r] ^{g_2} \ar@{-}[d]|{}
\ar@{}[dr]|{\eta_2} &  \ar@{-}[d]^{} 
\\ \ar@{-}[r] _{g'_1} & \ar@{-}[r] _{g'_2} & 
}}
 \,  = \,
\xybiglabels \vcenter{\xymatrix@M=0pt@=3pc@C=2pc{\ar@{-} [d] _{} \ar@{-} [rrr]^-{g_1g_2} & \ar@{} [dr]|-{\eta_1(g_1\rhd \eta_2) }& &\ar@{-} [d]^{} \\
\ar@{-} [rrr]_-{g'_1g'_2}  & &&}}
\qquad \qquad 
\xybiglabels \vcenter{\xymatrix@M=0pt@=3pc{\ar@{-} [d] _{} \ar@{-} [r]^{g} \ar@{} [dr]|{\eta} & \ar@{-} [d]^{} \\
\ar@{-} [r] |{g'} \ar@{-} [d]_{} \ar@{}[dr] |{\eta'} & \ar@{-} [d]^{} \\
\ar@{-} [r]_{g''}& }}  \,  = \,
\xybiglabels \vcenter{\xymatrix@M=0pt@=3pc{\ar@{-} [d] _{} \ar@{-} [r]^-{g}  \ar@{} [dr]|-{\eta' \eta} & \ar@{-} [d]^{} \\
\ar@{-} [r]_-{g''} &  }}
\end{equation*}
These operations are associative and satisfy the interchange law: given a 2 by 2 array of composable squares, the result of composing horizontally and then vertically is the same as composing vertically and then horizontally. Thus any rectangular array of squares has a unique evaluation as a single square.

Squares admit horizontal and vertical inverses, defined as follows:
\begin{equation*}
\xybiglabels \vcenter{\xymatrix@M=0pt@=3pc@C=1pc{\ar@{-} [d] _{} \ar@{-} [rrr]^-{g^{-1}} & \ar@{} [dr]|-{\eta^{-h}}& &\ar@{-} [d]^{} \\
\ar@{-} [rrr]_-{g'^{-1}}  & &&}} \, = \, 
\xybiglabels \vcenter{\xymatrix@M=0pt@=3pc@C=2pc{\ar@{-} [d] _{} \ar@{-} [rrr]^-{g^{-1}} & \ar@{} [dr]|-{g^{-1}\rhd \eta^{-1}}& &\ar@{-} [d]^{} \\
\ar@{-} [rrr]_-{g'^{-1}}  & &&}}
\qquad \qquad
\xybiglabels \vcenter{\xymatrix@M=0pt@=3pc@C=1pc{\ar@{-} [d] _{} \ar@{-} [rrr]^-{g'} & \ar@{} [dr]|-{\eta^{-v}}& &\ar@{-} [d]^{} \\
\ar@{-} [rrr]_-{g}  & &&}}  \, = \,
\xybiglabels \vcenter{\xymatrix@M=0pt@=3pc@C=1pc{\ar@{-} [d] _{} \ar@{-} [rrr]^-{g'} & \ar@{} [dr]|-{\eta^{-1}}& &\ar@{-} [d]^{} \\
\ar@{-} [rrr]_-{g}  & &&}} 
\end{equation*}
These inverses are appropriately both-sided (left/right for $\eta^{-h}$ and up/down for $\eta^{-v}$). Furthermore we will need the following properties for inverses of a horizontal or vertical composition:
\begin{eqnarray} 
\xybiglabels \vcenter{\xymatrix@M=0pt@=3pc@C=2pc{\ar@{-} [d] _{} \ar@{-} [rrr]^-{(g_1g_2)^{-1}} & \ar@{} [dr]|-{(\eta_1(g_1\rhd\eta_2))^{-h} }& &\ar@{-} [d]^{} \\
\ar@{-} [rrr]_-{(g'_1g'_2)^{-1}}  & &&}}
 &  = &
\xybiglabels \vcenter{\xymatrix @=3pc @W=0pc @M=0pc { \ar@{-}[r] ^{g_2^{-1}} \ar@{-}[d]
_{} \ar@{}[dr]|{\eta_2^{-h}} & \ar@{-}[r] ^{g_1^{-1}} \ar@{-}[d]|{}
\ar@{}[dr]|{\eta_1^{-h}} &  \ar@{-}[d]^{} 
\\ \ar@{-}[r] _{{g'_2}^{-1}} & \ar@{-}[r] _{{g'_1}^{-1}} & 
}}  \label{eq:hor-inverse-hor-prod}
\\
\xybiglabels \vcenter{\xymatrix@M=0pt@=3pc@C=2pc{\ar@{-} [d] _{} \ar@{-} [rrr]^-{g^{-1}} & \ar@{} [dr]|-{(\eta' \eta)^{-h}} & & \ar@{-} [d]^{} \\
\ar@{-} [rrr]_-{{g''}^{-1}} & && }}
 &  = &
\xybiglabels \vcenter{\xymatrix@M=0pt@=3pc@C=5pc{\ar@{-} [d] _{} \ar@{-} [r]^{g^{-1}} \ar@{} [dr]|{\eta^{-h}} & \ar@{-} [d]^{} \\
\ar@{-} [r] |{{g'}^{-1}} \ar@{-} [d]_{} \ar@{}[dr] |{{\eta'}^{-h}} & \ar@{-} [d]^{} \\
\ar@{-} [r]_{{g''}^{-1}}& }}   
\label{eq:hor-inverse-ver-prod} 
\\
\xybiglabels \vcenter{\xymatrix@M=0pt@=3pc@C=2pc{\ar@{-} [d] _{} \ar@{-} [rrr]^-{g''} & \ar@{} [dr]|-{(\eta' \eta)^{-v}} & & \ar@{-} [d]^{} \\
\ar@{-} [rrr]_-{{g}} & && }}
 &  = &
\xybiglabels \vcenter{\xymatrix@M=0pt@=3pc@C=5pc{\ar@{-} [d] _{} \ar@{-} [r]^{g''} \ar@{} [dr]|{\eta'^{-v}} & \ar@{-} [d]^{} \\
\ar@{-} [r] |{{g'}} \ar@{-} [d]_{} \ar@{}[dr] |{{\eta}^{-v}} & \ar@{-} [d]^{} \\
\ar@{-} [r]_{{g}}& }} \label{eq:ver-inverse-ver-prod} 
\end{eqnarray}

\section{Categories of Connections in Higher Gauge
  Theory}\label{sec:catconn}

Our objective is to describe higher connections and gauge transformations in relation to a given cell structure on $M$, so that connections become assignments of group elements to the 1- and 2-cells, and gauge transformations become assignments of group elements to the 0-cells. The connection assignments arise, in a differential geometric framework, from parallel transports obtained by integrating local connection 1- and 2-forms along the cells, taking account of the transitions between the open sets covering $M$. For a detailed account, see \cite{martinspickenii}. One of the underlying ideas is that these connection assignments constitute a discretization of a smooth connection, and the number of cells, although finite, can be as big as we like. On the other hand, for flat connections on simple manifolds, only a very small number of cells is required for a full characterization, and we will see in the examples of Section \ref{sec:geom-examples}   how such minimal discretizations can give a very efficient description.

Because of the algebraic relations between transports along oriented 1-cells and 2-cells, we introduce, in subsection  \ref{sec:Hdisc-conn-transp}, an additional structure, a bigon structure, on the 2-cells of $M$. This also reflects the fact that 
$M$ is a manifold and not just a topological space, i.e. it rules out some of the pathologies which may occur for the attaching maps of a cell complex in general.

\subsection{Discretized Higher Connections and the Category \texorpdfstring{$\cat{Conn}$}{} }
\label{sec:Hdisc-conn-transp}

We take our manifold $M$ to be endowed with a discretization $\cal D$, consisting of a finite cellular decomposition, a choice of orientation $O$ of the cells, and a choice of bigon structure $B$ on the 2-cells, to be described below. 
The cellular decomposition is a CW-decomposition, and we will frequently refer
to the $0$-, $1$- and $2$-cells as vertices, edges and faces, respectively. The corresponding
sets will be denoted $V$, $E$ and $F$. 

Let $\psi: S^1\rightarrow M$ be the attaching map of a 2-cell $f$ to the 1-skeleton of $M$, compatible with the 
positive orientation of the 2-cell. We say that the 2-cell has a bigon structure if the following conditions hold:

\begin{itemize}

\item the inverse image under $\psi$ of the 0-skeleton of $M$ consists of a finite subset $\{e^{i\theta_1}, \dots, e^{i\theta_n}\}\subset S^1\cong U(1)$, where
$0\leq \theta_1 < \theta_2 < \dots < \theta_n < 2\pi$.

\item  $\psi$ is injective onto an open 1-cell of $M$ for each of the collection of open arcs 
which constitute the complement in $S^1$ of the finite subset above.

\item one of the 0-cells in the boundary of the 2-cell is called the 0-source, denoted by $v$ in Figure \ref{fig:bigon-struc}, and one of the 0-cells is called the 0-target, denoted by $w$ in Figure \ref{fig:bigon-struc}. The 0-source and 0-target need not be distinct.

\item assuming the positive orientation of $f$ in Figure \ref{fig:bigon-struc}, the concatenations of 1-cells connecting the 0-source and the 0-target along the upper, respectively lower, boundary, are called the 1-source and 1-target of $f$ respectively, denoted by $e$ and $d$ in Figure \ref{fig:bigon-struc}.

\end{itemize}

\begin{figure}[h]
\begin{center}
\includegraphics[height=3cm]{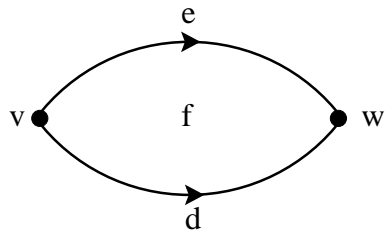}
\end{center}
\caption{Bigon structure on a $2$-cell}
\label{fig:bigon-struc}
\end{figure}

\begin{remark} Examples of 2-cells with bigon structure appear in Subsections \ref{subsec:S2-exp} and \ref{subsec:T2-exp}. 
A related construction is the notion of a {\em 2-lattice} defined in \cite[Def.~21]{bullivant-et-al}, which is also a CW-decomposition of $M$ with additional conditions. In particular, for a 2-lattice all cells of any dimension are endowed with a (single) basepoint in their boundary, and amongst other conditions on the attaching maps, all attaching maps of 3-cells are embeddings.

\end{remark}

Let $\G= (G,H,\rhd,\partial)$ be a crossed module. A discretized connection is  
an assignment of $G$-elements to the edges of $M$ and $H$ elements to the faces of $M$, 
respecting the bigon structure on the $2$-cells of $M$ as is made precise in the following definition.

\begin{definition}\label{def:catconn}
  The \textbf{category of connections}, $\cat{Conn} =
  \cat{Conn}(M, \G,{\cal D})$, is given as follows:
  \begin{itemize}
  \item \textbf{Objects} of $\cat{Conn}$
    consist of pairs of the form $(g,h)$, where $g : E \ra G,\, h : F \ra H$, subject to the condition:
\begin{equation*}
\xybiglabels \vcenter{\xymatrix@M=0pt@=3pc{\ar@{-} [d] _{} \ar@{-} [r]^{g(e)} \ar@{} [dr]|{h(f)} & \ar@{-} [d]^{} \\
\ar@{-} [r]_{g(d)}  & }}
\end{equation*}
for each face, 
i.e. $g(e)$ and $g(d)$ are not independent, but satisfy 
\begin{equation}\label{eq:conn-objects}
\partial(h(f))= g(d)g(e)^{-1}. 
\end{equation}
Since both $d$ and $e$ may be
composed of more than one edge, and each of these may be oriented from left to right or from right to left,
the value $g(e)$, and likewise $g(d)$, is determined by composing the $G$ elements assigned to the oriented
edges making up $e$, taken in order from left to right, and replacing the $G$ element by its inverse whenever the 
component edge is oriented from right to left.

    \item \textbf{Morphisms} of $\cat{Conn}$
consist of pairs $((g,h), \eta)$ where $(g,h)$ is an object and $\eta : E \ra H$. The source of $((g,h), \eta)$ is 
$(g,h)$ and the target of $((g,h), \eta)$ is $(g',h')$ given by:
\begin{equation} \label{eq:conn-morphisms-edge}
\xybiglabels \vcenter{\xymatrix@M=0pt@=3pc{\ar@{-} [d] _{} \ar@{-} [r]^{g(e)} \ar@{} [dr]|{\eta(e)} & \ar@{-} [d]^{} \\
\ar@{-} [r]_{g'(e)}  & }}
\end{equation}
for each edge $e$, and 
\begin{equation}
\xybiglabels \vcenter{\xymatrix@M=0pt@=3pc{\ar@{-} [d] _{} \ar@{-} [r]^{g(e)} \ar@{} [dr]|{h(f)} & \ar@{-} [d]^{} \\
\ar@{-} [r] |{g(d)} \ar@{-} [d]_{} \ar@{}[dr] |{\eta(d)} & \ar@{-} [d]^{} \\
\ar@{-} [r]_{g'(d)}& }} 
 \,  = \,
\xybiglabels \vcenter{\xymatrix@M=0pt@=3pc{\ar@{-} [d] _{} \ar@{-} [r]^{g(e)} \ar@{} [dr]|{\eta(e)} & \ar@{-} [d]^{} \\
\ar@{-} [r] |{g'(e)} \ar@{-} [d]_{} \ar@{}[dr] |{h'(f)} & \ar@{-} [d]^{} \\
\ar@{-} [r]_{g'(d)}& }} 
\label{eq:conn-morphisms}
\end{equation}
for each face $f$. Again, since both $d$ and $e$ may be
composed of more than one edge, and each of these may be oriented from left to right or from right to left,
the value $\eta(e)$, and likewise $\eta(d)$, is determined by multiplying horizontally the squares (\ref{eq:conn-morphisms-edge}) assigned to the oriented
edges making up $e$, taken in order from left to right, and replacing the square by its horizontal inverse whenever the 
component edge is oriented from right to left.

 \item \textbf{Composition} of morphisms is defined by 
  \begin{equation}\label{eq:conn-composition}
    ((g',h'), \eta')\circ   ((g,h), \eta) =  ((g,h), \eta'\eta)
  \end{equation}
where $(\eta'\eta)(e) = \eta'(e)\eta(e)$ for each edge. 
 \item \textbf{Identities} are given, for each object $(g,h)$, by
   \begin{equation}\label{eq:conn-identities}
    id_{(g,h)} =  ((g,h),1)
  \end{equation}

\end{itemize}
\end{definition}

\begin{remark}
We note that (\ref{eq:conn-morphisms}) can be rewritten to give a formula for the square with $h'(f)$, using the vertical inverse of the square above it in (\ref{eq:conn-morphisms}). We return to this in subsection \ref{subsec:S2-exp}.
\label{rem:h'-formula}
\end{remark}

\begin{remark}
We observe that (\ref{eq:conn-morphisms}) may be viewed from a 3D perspective as the equation for a commuting bigon cylinder - see Figure \ref{fig:comm-bigon-cyl}. We return to discussing such higher-dimensional perspectives in Remark \ref{rem:3Daction} and in the examples of Section \ref{sec:geom-examples}.
\label{rem:bigon-cyl}
\end{remark}

\begin{remark}
\label{rem:fullgt1}
The morphisms in $\cat{Conn}$ correspond to a special case of the {\em full gauge transformations} of  \cite[Def. 86, Fig. 7]{bullivant-et-al}). The latter are given in general by assignments of $G$ elements to vertices and $H$ elements to edges, and our morphisms correspond to trivializing the assignments to vertices (setting them all equal to $1_G$). We will return to this point in Remark \ref{rem:fullgt2}, when we introduce our notion of gauge transformations.
\end{remark}

\begin{figure}[htbp] 
\centerline{\relabelbox 
\epsfysize 6cm
\epsfig{file=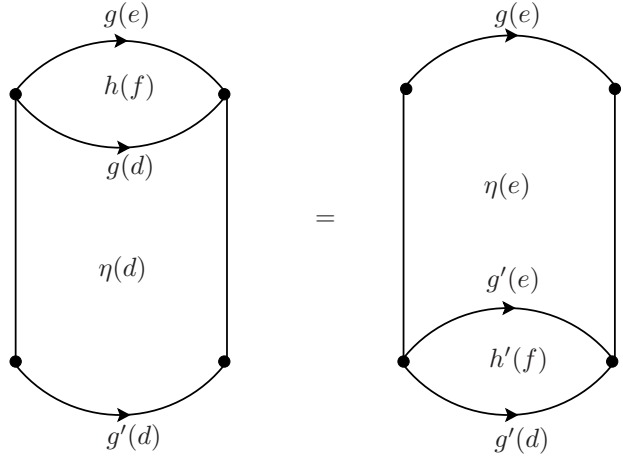,height=6cm}
\relabel{eta(d)}{$\eta(d)$}
\relabel{eta(e)}{$\eta(e)$}
\relabel{g(e)}{$g(e)$}
\relabel{g(d)}{$g(d)$}
\relabel{k(e)}{$g(e)$}
\relabel{=}{$=$}
\relabel{g'(d)}{$g'(d)$}
\relabel{g'(e)}{$g'(e)$}
\relabel{k'(d)}{$g'(d)$}
\relabel{h(f)}{$h(f)$}
\relabel{h'(f)}{$h'(f)$}
\endrelabelbox}
\caption{\label{fig:comm-bigon-cyl} 3D perspective on (\ref{eq:conn-morphisms}) - the composition of the two visible faces of the bigon cylinder equals the composition of the two hidden faces }
\end{figure}

\begin{theorem} The category $\cat{Conn}$ is well-defined.
\label{thm:conn}
\end{theorem}
\begin{proof}
The main point to verify is that composition is well-defined. This is shown by combining
condition (\ref{eq:conn-morphisms}) for the morphism $((g,h), \eta)$ with the corresponding condition 
for the morphism $((g',h'), \eta')$:
\begin{equation*}
\xybiglabels \vcenter{\xymatrix@M=0pt@=3pc{\ar@{-} [d] _{} \ar@{-} [r]^{g'(e)} \ar@{} [dr]|{h'(f)} & \ar@{-} [d]^{} \\
\ar@{-} [r] |{g'(d)} \ar@{-} [d]_{} \ar@{}[dr] |{\eta'(d)} & \ar@{-} [d]^{} \\
\ar@{-} [r]_{g''(d)}& }} 
 \,  = \,
\xybiglabels \vcenter{\xymatrix@M=0pt@=3pc{\ar@{-} [d] _{} \ar@{-} [r]^{g'(e)} \ar@{} [dr]|{\eta'(e)} & \ar@{-} [d]^{} \\
\ar@{-} [r] |{g''(e)} \ar@{-} [d]_{} \ar@{}[dr] |{h''(f)} & \ar@{-} [d]^{} \\
\ar@{-} [r]_{g''(d)}& }}
\end{equation*}
and multiplying these equations vertically after inserting the square 
$$
\xybiglabels \vcenter{\xymatrix@M=0pt@=3pc{\ar@{-} [d] _{} \ar@{-} [r]^{g'(d)} \ar@{} [dr]|{h'(f)^{-v}} & \ar@{-} [d]^{} \\
\ar@{-} [r]_{g'(e)}  & }}
$$
between them. Cancelling the $h'(f)$ square with this vertical inverse on both sides of the equation, and composing the $\eta$
and $\eta'$ squares vertically, gives equation (\ref{eq:conn-morphisms}) for
the composite morphism $((g,h), \eta'\eta)$, ensuring that it has the correct target:
\begin{equation*}
\xybiglabels \vcenter{\xymatrix@M=0pt@=3pc@C=5pc{\ar@{-} [d] _{} \ar@{-} [r]^{g(e)} \ar@{} [dr]|{h(f)} & \ar@{-} [d]^{} \\
\ar@{-} [r] |{g(d)} \ar@{-} [d]_{} \ar@{}[dr] |{(\eta'\eta)(d)} & \ar@{-} [d]^{} \\
\ar@{-} [r]_{g''(d)}& }} 
 \,  = \,
\xybiglabels \vcenter{\xymatrix@M=0pt@=3pc@C=5pc{\ar@{-} [d] _{} \ar@{-} [r]^{g(e)} \ar@{} [dr]|{(\eta'\eta)(e)} & \ar@{-} [d]^{} \\
\ar@{-} [r] |{g''(e)} \ar@{-} [d]_{} \ar@{}[dr] |{h''(f)} & \ar@{-} [d]^{} \\
\ar@{-} [r]_{g''(d)}& }} 
\end{equation*}
Composition is clearly associative and the identity morphisms obviously have the right properties.

\end{proof}

To conclude this subsection, we note that the squares (\ref{eq:conn-morphisms-edge}) have vertical inverses
\begin{equation} \label{eq:conn-morphisms-edge-inverse}
\xybiglabels \vcenter{\xymatrix@M=0pt@=3pc{\ar@{-} [d] _{} \ar@{-} [r]^{g'(e)} \ar@{} [dr]|{\eta(e)^{-v}} & \ar@{-} [d]^{} \\
\ar@{-} [r]_{g(e)}  & }}
\end{equation}
and hence we have:

\begin{lemma}
The category $\cat{Conn}$ is a groupoid.
\end{lemma}

\begin{proof}
The properties of the vertical inverse ensure that each morphism of $\cat{Conn}$ has an inverse given by
(\ref{eq:conn-morphisms-edge-inverse}). This is well defined, since the condition corresponding to 
(\ref{eq:conn-morphisms}) is derived from (\ref{eq:conn-morphisms}) itself by composing vertically on both sides with
(\ref{eq:conn-morphisms-edge-inverse}) at the top, and composing with the corresponding square (\ref{eq:conn-morphisms-edge-inverse}) for $d$ at the bottom.
\end{proof}

\subsection{The 2-Group \texorpdfstring{$\cat{Gauge}$}{} and its Action on \texorpdfstring{$\cat{Conn}$}{}}
\label{sec:2group-action}

With a 2-group $\G$ and discretized manifold $(M, {\cal D})$ as before, the 2-group of gauge transformations is described by data assigned only to the 0-cells $V$, as follows.
\begin{definition}
\label{def:2gp-gauge}
The 2-group of gauge transformations, $\cat{Gauge}$, regarded as a categorical group, is given as follows:
  \begin{itemize}
  \item \textbf{Objects} are the set of  maps $\gamma : V \ra G$

 \item \textbf{Morphisms} are the set of pairs $(\gamma,\chi)$ where $\gamma$ is an object and $\chi : V \ra H$. The source and target of $(\gamma,\chi)$ are $s(\gamma,\chi)=\gamma$ and $t(\gamma,\chi)=\gamma'$, where
\begin{equation}
\xybiglabels \vcenter{\xymatrix@M=0pt@=3pc{\ar@{-} [d] _{} \ar@{-} [r]^{\gamma(v)} \ar@{} [dr]|{\chi(v)} & \ar@{-} [d]^{} \\
\ar@{-} [r]_{\gamma'(v)}  & }}
\label{eq:gaugemor}
\end{equation}
for each $v\in V$. 

 \item \textbf{Composition} of morphisms $(\gamma,\chi)$ and $(\gamma',\chi')$ is given pointwise by vertical composition of the squares (\ref{eq:gaugemor}). 

\item \textbf{Identities}  are the morphisms with $\chi(v)=1,\, \forall v$. 

\item The \textbf{monoidal structure} of the categorical group is given  pointwise by horizontal composition of the squares (\ref{eq:gaugemor}).  

  \end{itemize}
\end{definition}

It is clear that $\cat{Gauge}$ is a well-defined categorical group, since it the product over $V$ of the categorical group $\G$.

In \cite{morton-picken-i} we defined the general notion of the action of a 2-group on a 
category\footnote{Note that the action of a 2-group $\G$ on a category could also be regarded
as the action of the associated double groupoid 
\cite{brown-higgins-sivera,martinspickenii, morton-picken-i} ${\cal D}(\G)$ 
on the category.
In this context we would like to mention that there is a notion of action 
of a double groupoid on morphisms of groupoids \cite[Def. 1.5]{brown-mackenzie}, but this is in the sense of a groupoid action
on a map \cite[Def. 2.1]{higgins-mackenzie}, so a somewhat different perspective to ours. We are grateful to Ronnie Brown
for drawing our attention to this point.}. 
This can be done at the 2-categorical or at the categorical level; here we choose the categorical level, which is Definition 3.3 of \cite{morton-picken-i}. In short, a (strict) action is a functor such that an action diagram commutes and a unit condition holds. 

 \begin{definition}
  A strict action of a categorical group $\G$ on a category $\C$ is a
  functor $\hat{\Phi}:\G\times \C \rightarrow \C$ 
  satisfying the action square diagram 
  in $\cat{Cat}$ (strictly):
\begin{equation}\label{eq:catactioncondition-diag}
  \xymatrix@C=+5pc{
    \G \times \G \times \C \ar[r]^{\otimes \times Id_{\C}} \ar[d]_{Id_{\G} \times \hat{\Phi}} & \G \times \C \ar[d]^{\hat{\Phi}} \\
    \G \times \C \ar[r]_{\hat{\Phi}} & \C
  }
\end{equation}
and the unit condition 
\begin{equation}\label{eq:Phi-unit-condition}
 \hat{\Phi}(1,x) =x, \quad 
\hat{\Phi}(id_1,f) = f  
\end{equation}
for all objects $x$ and morphisms $f$ of $\C$
\label{def:2grp_action-Phi-hat}
\end{definition}
\begin{remark}
Note that in \cite{morton-picken-i} we omitted the unit condition on morphisms. Definition \ref{def:2grp_action-Phi-hat} may be regarded as defining  a ``categorified action'', since at the object level  (\ref{eq:catactioncondition-diag}) and (\ref{eq:Phi-unit-condition}) translate to the usual conditions $g_1.(g_2.x)=(g_1g_2).x$ and $1.x=x$ for a group action (where $g.x$ denotes $\hat{\Phi}(g,x)$).
\end{remark}

For the present case, we define the following action.

 \begin{definition}
The 2-group of gauge transformations acts on the groupoid of connections by:
\begin{itemize}
  \item on objects $\hat{\Phi}(\gamma, (g,h))=(\gamma . g, \gamma . h)$ where, for any edge $e\in E$ from $v$ to $w$,
  \begin{equation}
\label{eq:Phi-ob-edge}
(\gamma.g)(e) := \gamma(v)g(e)\gamma(w)^{-1},
  \end{equation}
and for any face with bigon structure $f\in F$:
\begin{equation}\label{eq:Phi-ob-face}
\xybiglabels \vcenter{\xymatrix@M=0pt@=3pc@C=4pc{\ar@{-} [d] \ar@{-} [r]^{(\gamma. g)(e)} \ar@{} [dr]|{(\gamma.h)(f)} & \ar@{-} [d] \\
\ar@{-} [r]_{(\gamma. g)(d)}  & }}
\, = \,
\xybiglabels \vcenter{\xymatrix @=3pc @W=0pc @M=0pc { \ar@{-}[r] ^{\gamma(v)} \ar@{-}[d]_{} \ar@{}[dr]|{} & \ar@{-}[r] ^{g(e)} \ar@{-}[d]|{}
\ar@{}[dr]|{h(f)} & \ar@{-}[r] ^{\gamma(w)^{-1}} \ar@{-}[d]^{}  \ar@{}[dr]|{} & \ar@{-}[d]^{}
\\ \ar@{-}[r] _{\gamma(v)} & \ar@{-}[r] _{g(d)} & \ar@{-}[r] _{\gamma(w)^{-1}} &
}}  
\end{equation}

  \item on morphisms $\hat{\Phi}((\gamma,\chi), ((g,h),\eta))=((\gamma.g, \gamma.h), (\gamma,\chi).\eta)$ where, for any edge $e\in E$ from $v$ to $w$,
  \begin{equation}
\label{eq:Phimor}
    \xybiglabels \vcenter{\xymatrix@M=0pt@=3pc@C=5pc{\ar@{-} [d]  \ar@{-} [r]^{(\gamma.g)(e)} \ar@{} [dr]|{((\gamma,\chi).\eta)(e)} & \ar@{-} [d] \\ \ar@{-} [r]_{(\gamma.g)'(e)}  & }}
    \,  = \,
    \xybiglabels \vcenter{\xymatrix @=3pc @W=0pc @M=0pc { \ar@{-}[r] ^{\gamma(v)} \ar@{-}[d]_{} \ar@{}[dr]|{\chi(v)} & \ar@{-}[r] ^{g(e)} \ar@{-}[d]|{}
        \ar@{}[dr]|{\eta(e)} & \ar@{-}[r] ^{\gamma(w)^{-1}} \ar@{-}[d]^{}  \ar@{}[dr]|{\chi(w)^{-h}} & \ar@{-}[d]^{}
        \\ \ar@{-}[r] _{\gamma'(v)} & \ar@{-}[r] _{g'(e)} & \ar@{-}[r] _{\gamma'(w)^{-1}} &
      }}.
  \end{equation}

\end{itemize}
\label{def:Phi-hat-for-Conn+Gauge}
\end{definition}

\begin{remark} 
Equation (\ref{eq:Phimor}) could also be viewed from a 3D perspective as a commuting 3-cube. An instance of this is displayed in Figure  \ref{fig:3D-circle} for the example when $M$ is a circle.

In \cite{morton-picken-i} we used the symbol $\act$ as a shorthand for the action in the general case (see Def. 3.5  of \cite{morton-picken-i}). Thus for instance the unit condition (\ref{eq:Phi-unit-condition}) may be written: $1\act x=x,\, id_1 \act f = f $. Here we are writing the action for our specific case using the symbol $.$, i.e. $(\gamma.g,\gamma.h)$ etc.
\label{rem:3Daction}
\end{remark}

\begin{remark}
The relation between our 2-group of gauge transformations $\cat{Gauge}$ and the full gauge transformations of \cite[Def. 86, Fig. 7, Rem. 89]{bullivant-et-al} is as follows 
(recall Remark \ref{rem:fullgt1}). Our action of $\cat{Gauge}$ on $\cat{Conn}$ at the object level is given by assignments of $G$ elements to the vertices of $M$, and is thus a special case of full gauge transformations, obtained by trivializing the assignments of $H$ elements to edges.  Thus from our perspective, full gauge transformations combine the morphisms of $\cat{Conn}$ and the action of $\cat{Gauge}$ on $\cat{Conn}$ at the object level. For the purpose of their study of topological phases of matter in 3+1 dimensions, Bullivant et al do not need to introduce the higher gauge transformations corresponding in our terms to the action of  $\cat{Gauge}$ on $\cat{Conn}$ at the morphism level, but they do suggest an algebraic framework for these as 2-fold homotopies between crossed module homotopies
 \cite[Rem. 94]{bullivant-et-al}.
\label{rem:fullgt2}
\end{remark}

The central result of this article is:

\begin{theorem} The action $\hat{\Phi}$ of the previous definition is a well-defined action of $\cat{Gauge}$ on $\cat{Conn}$.
\label{thm:main-action}
\end{theorem}
We will divide the proof up into smaller lemmas. Apart from showing that $\hat{\Phi}$ is well-defined on morphisms, we need to show that it is a functor, that it satisfies the action square diagram and that it satisfies the unit condition.
\begin{lemma}
\label{lem:well-defined}
$\hat{\Phi}$  is well-defined on morphisms.
\end{lemma}
\begin{proof}
We need to show that $\hat{\Phi}((\gamma,\chi), ((g,h),\eta))$ is a well-defined morphism in $\cat{Conn}$, satisfying (\ref{eq:conn-morphisms}) for each face, i.e.:
\begin{equation*}
\xybiglabels \vcenter{\xymatrix@M=0pt@=3pc@C=5pc{\ar@{-} [d] _{} \ar@{-} [r]^{(\gamma.g)(e)} \ar@{} [dr]|{(\gamma.h)(f)} & \ar@{-} [d]^{} \\
\ar@{-} [r] |{(\gamma.g)(d)} \ar@{-} [d]_{} \ar@{}[dr] |{((\gamma,\chi).\eta)(d)} & \ar@{-} [d]^{} \\
\ar@{-} [r]_{(\gamma.g)'(d)}& }} 
 \,  = \,
\xybiglabels \vcenter{\xymatrix@M=0pt@=3pc@C=5pc{\ar@{-} [d] _{} \ar@{-} [r]^{(\gamma.g)(e)} \ar@{} [dr]|{((\gamma,\chi).\eta)(e)} & \ar@{-} [d]^{} \\
\ar@{-} [r] |{(\gamma.g)'(e)} \ar@{-} [d]_{} \ar@{}[dr] |{(\gamma.h)'(f)} & \ar@{-} [d]^{} \\
\ar@{-} [r]_{(\gamma.g)'(d)}& }} 
\end{equation*}

This is the horizontal composition of the equation
$$
\xybiglabels \vcenter{\xymatrix @=3pc @W=0pc @M=0pc @C=4pc { \ar@{-}[r] ^{\gamma(v)} \ar@{-}[d]_{} \ar@{}[dr]|{} & \ar@{-}[r] ^{g(e)} \ar@{-}[d]|{} \ar@{}[dr]|{h(f)} & 
\ar@{-}[r] ^{\gamma(w)^{-1}} \ar@{-}[d]^{}  \ar@{}[dr]|{} & \ar@{-}[d]^{}
\\ \ar@{-}[r] |{\gamma(v)} \ar@{-}[d]_{} \ar@{}[dr]|{\chi(v)} & \ar@{-}[r] |{g(d)} \ar@{-}[d]|{} \ar@{}[dr]|{\eta(d)}  & \ar@{-}[r] |{\gamma(w)^{-1}} \ar@{-}[d]^{}  \ar@{}[dr]|{\chi(w)^{-h}} & \ar@{-}[d]^{}
\\ \ar@{-}[r] _{\gamma'(v)} & \ar@{-}[r] _{g'(d)} & \ar@{-}[r] _{\gamma'(w)^{-1}} &
}}
\, = \,
\xybiglabels \vcenter{\xymatrix @=3pc @W=0pc @M=0pc @C=4pc { \ar@{-}[r] ^{\gamma(v)} \ar@{-}[d]_{} \ar@{}[dr]|{\chi(v)} & \ar@{-}[r] ^{g(e)} \ar@{-}[d]|{} \ar@{}[dr]|{\eta(e)} & 
\ar@{-}[r] ^{\gamma(w)^{-1}} \ar@{-}[d]^{}  \ar@{}[dr]|{\chi(w)^{-h}} & \ar@{-}[d]^{}
\\ \ar@{-}[r] |{\gamma'(v)} \ar@{-}[d]_{} \ar@{}[dr]|{} & \ar@{-}[r] |{g'(e)} \ar@{-}[d]|{} \ar@{}[dr]|{h'(f)}  & \ar@{-}[r] |{\gamma'(w)^{-1}} \ar@{-}[d]^{}  \ar@{}[dr]|{} & \ar@{-}[d]^{}
\\ \ar@{-}[r] _{\gamma'(v)} & \ar@{-}[r] _{g'(d)} & \ar@{-}[r] _{\gamma'(w)^{-1}} &
}}
$$
which holds since the three vertical compositions are equal on either side of the equation, the outer ones because of the properties of vertical identity squares, and the middle one since it is (\ref{eq:conn-morphisms}) for the morphism $((g,h),\eta)$.
\end{proof}

\begin{lemma} $\hat{\Phi}$  is functorial.
\label{lem:functor}
\end{lemma}
\begin{proof}
We need to show that $\hat{\Phi}$ preserves compositions and identities. For compositions this means showing the equation:
$$
\hat{\Phi}((\gamma',\chi'), ((g',h'),\eta')) \circ  \hat{\Phi}((\gamma,\chi), ((g,h),\eta)) 
= \hat{\Phi}((\gamma,\chi'\chi), ((g,h),\eta'\eta)).
$$
This holds, since the left hand side corresponds to:
$$
\xybiglabels \vcenter{\xymatrix @=3pc @W=0pc @M=0pc @C=4pc { \ar@{-}[r] ^{\gamma(v)} \ar@{-}[d]_{} \ar@{}[dr]|{\chi(v)} & \ar@{-}[r] ^{g(e)} \ar@{-}[d]|{} \ar@{}[dr]|{\eta(e)} & 
\ar@{-}[r] ^{\gamma(w)^{-1}} \ar@{-}[d]^{}  \ar@{}[dr]|{\chi(w)^{-h}} & \ar@{-}[d]^{}
\\ \ar@{-}[r] |{\gamma'(v)} \ar@{-}[d]_{} \ar@{}[dr]|{\chi'(v)} & \ar@{-}[r] |{g'(e)} \ar@{-}[d]|{} \ar@{}[dr]|{\eta'(e)}  & \ar@{-}[r] |{\gamma'(w)^{-1}} \ar@{-}[d]^{}  \ar@{}[dr]|{\chi'(w)^{-h}} & \ar@{-}[d]^{}
\\ \ar@{-}[r] _{\gamma''(v)} & \ar@{-}[r] _{g''(d)} & \ar@{-}[r] _{\gamma''(w)^{-1}} &
}}
$$
and multiplying vertically gives the right hand side. 

For identities we need to show
$$
\hat{\Phi}((\gamma,1), ((g,h),1)) = ((\gamma . g, \gamma . h), 1)
$$
which is obvious, setting $\chi(v)$, $\eta(e)$ and $\chi(w)$ equal to 1 in (\ref{eq:Phimor}).
\end{proof}

\begin{lemma} $\hat{\Phi}$  satisfies the action square diagram (\ref{eq:catactioncondition-diag}).
\label{lem:action-square}
\end{lemma}
\begin{proof}
The commutativity of the action diagram (\ref{eq:catactioncondition-diag}) at the object level corresponds to the statement:
$$
((\tilde{\gamma} \gamma).g, (\tilde{\gamma} \gamma).h))= (\tilde{\gamma}.(\gamma.g), \tilde{\gamma}.(\gamma.h))
$$
for any $\tilde{\gamma}, \gamma: V\ra G$. For functions of edges $g:E\ra G$ this is immediate from (\ref{eq:Phi-ob-edge}). For functions of faces $h:F\ra H$, this follows from considering the array:
$$
\xybiglabels \vcenter{\xymatrix @=3pc @W=0pc @M=0pc { \ar@{-}[r] ^{\tilde{\gamma}(v)} \ar@{-}[d]_{} \ar@{}[dr]|{} & \ar@{-}[r] ^{\gamma(v)} \ar@{-}[d]_{} \ar@{}[dr]|{} & 
\ar@{-}[r] ^{g(e)} \ar@{-}[d]|{} \ar@{}[dr]|{h(f)} & 
\ar@{-}[r] ^{\gamma(w)^{-1}} \ar@{-}[d]^{}  \ar@{}[dr]|{} & \ar@{-}[r] ^{\tilde{\gamma}(w)^{-1}} \ar@{-}[d]^{}  \ar@{}[dr]|{} & \ar@{-}[d]^{}
\\ \ar@{-}[r] _{\tilde{\gamma}(v)} & \ar@{-}[r] _{\gamma(v)} & \ar@{-}[r] _{g(d)} & \ar@{-}[r] _{\gamma(w)^{-1}} & \ar@{-}[r] _{\tilde{\gamma}(w)^{-1}} &
}}
$$
and either composing horizontally the two squares on the left and the two squares on the right, or composing horizontally the three squares in the middle.

To show the commutativity of the action diagram (\ref{eq:catactioncondition-diag}) at the morphism level, we consider the monoidal product of two morphisms in $\cat{Gauge}$:
\begin{equation}
\xybiglabels \vcenter{\xymatrix @=3pc @W=0pc @M=0pc { \ar@{-}[r] ^{\tilde{\gamma}} \ar@{-}[d]_{} \ar@{}[dr]|{\tilde{\chi}} & \ar@{-}[r] ^{\gamma} \ar@{-}[d]|{}
\ar@{}[dr]|{\chi} & \ar@{-}[d]^{}
\\ \ar@{-}[r] _{\tilde{\gamma}'} & \ar@{-}[r] _{\gamma'} & 
}}
\label{eq:monoidal-prod}
\end{equation}
This commutativity then corresponds to the statement:
$$
(((\tilde{\gamma}\gamma).g,(\tilde{\gamma}\gamma).h),(\tilde{\chi}(\tilde{\gamma}\rhd\chi)).\eta)= ((\tilde{\gamma}.(\gamma.g),\tilde{\gamma}.(\gamma.h)), (\tilde{\gamma},\tilde{\chi}).((\gamma, \chi).\eta))
$$
The equality of the first component has already been established. The equality of the second component follows from considering the array:
$$
\xybiglabels \vcenter{\xymatrix @=3pc @W=0pc @M=0pc { \ar@{-}[r] ^{\tilde{\gamma}(v)} \ar@{-}[d]_{} \ar@{}[dr]|{\tilde{\chi}(v)} & \ar@{-}[r] ^{\gamma(v)} \ar@{-}[d]_{} \ar@{}[dr]|{\chi(v)} & 
\ar@{-}[r] ^{g(e)} \ar@{-}[d]|{} \ar@{}[dr]|{\eta(e)} & 
\ar@{-}[r] ^{\gamma(w)^{-1}} \ar@{-}[d]^{}  \ar@{}[dr]|{\chi(w)^{-h}} & \ar@{-}[r] ^{\tilde{\gamma}(w)^{-1}} \ar@{-}[d]^{}  \ar@{}[dr]|{\tilde{\chi}(w)^{-h}} & \ar@{-}[d]^{}
\\ \ar@{-}[r] _{\tilde{\gamma}'(v)} & \ar@{-}[r] _{\gamma'(v)} & \ar@{-}[r] _{g'(e)} & \ar@{-}[r] _{\gamma'(w)^{-1}} & \ar@{-}[r] _{\tilde{\gamma}'(w)^{-1}} &
}}
$$
and either composing horizontally the two squares on the left and the two squares on the right, and using properties of the horizontal inverse, or composing horizontally the three squares in the middle.
\end{proof}

\begin{lemma} $\hat{\Phi}$  satisfies the unit condition.
\label{lem:unit-cond}
\end{lemma}
\begin{proof}
On objects, the unit condition is 
$$
(1.g,1.h)=(g,h)
$$
where $1:V\ra G$ is such that $1(v)=1, \forall v\in V$. This clearly holds, putting $\gamma=1$ in (\ref{eq:Phi-ob-edge}) and (\ref{eq:Phi-ob-face}).

On morphisms the unit condition is
$$
\hat{\Phi}({\rm id}_1,((g,h),\eta)) = (1,1).((g,h),\eta)=((g,h),\eta)
$$
where $(1,1):V\ra G\times H$ is such that $(1,1)(v)=(1,1),\, \forall v\in V$. Again this trivially holds putting 
$(\gamma, \chi)=(1,1)$ in (\ref{eq:Phimor}).
\end{proof}

\subsection{Effect of Changes in Discretization}
\label{sec:ch-discr}

It is worth noting that our definition of the category of connections
makes sense only relative to a particular choice of discretization
$\cal{D}$. Nevertheless, this notion of a connection does capture some
of the information contained in a connection in the sense of differential
geometry. Indeed, for a flat connection, it will contain all the
relevant information. The situation at hand is somewhat analogous to other
situations where one makes arbitrary choices such as a choice of local
coordinates, or fixing a gauge, in order to capture a geometric
structure, and describes how the result transforms under a change in
this choice. Aside from this section, we will continue to work in the
setting where we have already made a specific choice, but
in order to clarify the link to the usual differential-geometric
picture, we will describe here how the construction of $\cat{Conn}$
would vary with changes to $\cal{D}$.

In particular, given $(M, \cal{D}, \G)$, the definition of $\cat{Conn}$
depends in the first instance on the cell decomposition which is part
of $\cal{D}$, and we will observe
that changing the choice of orientations $O$ on edges and faces, 
and changing the choice of bigon structures $B$
on faces, gives a straightforward, but non-identity, isomorphism
between the respective categories $\cat{Conn}$ associated to this 
change. We will
formulate some conjectures concerning the effect of more substantial
changes to $\cal{D}$.

A connection in higher gauge theory based on a Lie 2-group $\G$
given by the crossed module $(G,H,\rhd,\partial)$ may be described locally in
terms of a 1-form $A$, valued in $Lie(G)$, the Lie algebra of
$G$, and a 2-form $B$, valued in $Lie(H)$, the Lie algebra of
$H$.  There are also $Lie (H)$-valued transition 1-forms and
transition functions valued in $G$ and $H$, and a global description
in terms of parallel transport, to which we will return in our
forthcoming work \cite{morton-picken-iii}. For further discussion of
higher gauge theory from this point of view, the reader may consult a
variety of works on the subject, such as \cite{basch-hgt,
  schreiberwaldorfii,  martinspickenii, soncinizucchini}.

The relation between such a situation and the discrete description
given here is as follows. Recall that for ordinary gauge theory with
Lie group $G$, a $Lie (G)$-valued connection 1-form gives $G$-valued holonomies
or parallel transports along paths (for flat connections, this is
determined by the homotopy class of the path).  These tell how to
transport a fibre $F$ which carries a $G$-action along a path $\gamma$
from $x$ to $y$
which, after fixing a basepoint in $F_x$ and $F_y$, determines a
correspondence between the two fibres.

A 2-group-valued connection, on the other hand, gives parallel transports for
both paths and surfaces.  Given a homotopy of paths $\Gamma$,
understood as a family of paths $\gamma_t$, with $t \in [0,1]$, which sweeps
out a surface, the holonomy
\begin{equation}
  h = hol(\Gamma)
\end{equation}
for that surface can be seen as a 2-morphism $(g,h)$ in $\G$,
regarded as a 2-category, relating
\begin{equation}
  g = hol(\gamma_0)
\end{equation}
and
\begin{equation}\label{eq:holtarget}
  \partial(h) g = hol(\gamma_1)
\end{equation}

These holonomies are obtained from the local $\G$-connection by
integrating the connection forms over a path or surface, respectively,
taking into account the transition 1-forms and transition functions -
for a detailed description see \cite{martinspickenii}.

Then, as in ordinary (group-valued) gauge theory, gauge
transformations take one connection to another.  Gauge transformations can be
expressed locally as $G$-valued functions and $Lie(H)$-valued
1-forms. The latter give $H$-valued holonomies on paths after
integration.  A new feature of higher gauge theory is that there are
higher gauge transformations as well, given by $H$-valued functions.

In our discrete setting, given a cell structure $\cal{D}$, we can
obtain assignments $g(e)$ by considering an oriented edge $e$ as an
equivalence class of paths $\gamma_0$ consisting of all
parametrizations of $e$, and letting $g(e) =
hol(\gamma_0)$. Similarly, to a face $f$ equipped with a bigon
structure, there corresponds an equivalence class of homotopies of paths 
$\Gamma$,
namely those taking its 1-source to its 1-target and having image
$f$. Thus, we may take $h(f) = hol(\Gamma)$. Likewise for the 1-target 
we have $g(d)=hol(\gamma_1)$, and thus we have a match with the 
assignments in Def. \ref{def:catconn} which are the objects of the category
$\cat{Conn}$.

Now consider a manifold $M$ with two choices of
discretization $\cal{D}$ and $\cal{D'}$ which share the same cell structure and
differ only by the choice of orientations $O$ and $O'$,
and bigon structures $B$ and $B'$. Then we want to
understand the relation between $\cat{Conn}(M, \cal{D},\G)$ and
$\cat{Conn}(M, \cal{D'},\G)$.

We consider the following four types of changes.

\begin{itemize}
\item[1)] change of edge orientation

Suppose $\cal{D}'$ has edge $e_i$ of $\cal{D}$ replaced by $\overline{e}_i$, the oppositely oriented edge.
Then the change in the objects of $\cat{Conn}$ is given by 
$(g,h)\mapsto (\tilde{g}, \tilde{h})$ where on edges and faces respectively we have:
\begin{equation}\label{eq:edgeflip}
g(e_i)\mapsto \tilde{g}(\overline{e}_i)= g(e_i)^{-1}, \qquad 
\xybiglabels \vcenter{\xymatrix@M=0pt@=3pc{\ar@{-} [d] _{} \ar@{-} [r]^{g(e)} \ar@{} [dr]|{h(f)} & \ar@{-} [d]^{} \\
\ar@{-} [r]_{g(d)}  & }}
\mapsto
\xybiglabels \vcenter{\xymatrix@M=0pt@=3pc{\ar@{-} [d] _{} \ar@{-} [r]^{\tilde{g}(e)} \ar@{} [dr]|{\tilde{h}(f)} & \ar@{-} [d]^{} \\
\ar@{-} [r]_{\tilde{g}(d)}  & }} = 
\xybiglabels \vcenter{\xymatrix@M=0pt@=3pc{\ar@{-} [d] _{} \ar@{-} [r]^{g(e)} \ar@{} [dr]|{h(f)} & \ar@{-} [d]^{} \\
\ar@{-} [r]_{g(d)}  & }}
\end{equation}
Indeed the assignments to faces are unchanged due to the orientation conventions 
given below (\ref{eq:conn-objects}) in Def. \ref{def:catconn}.

\item[2)] change of face orientation by vertical inversion

This change of orientation simultaneously affects the bigon structure and is given by 
the replacement of a 2-cell of $\cal{D}$ by a corresponding 2-cell of $\cal{D}'$ with 
the 1-source and 1-target exchanged:

\begin{figure}[htbp] 
\centerline{\relabelbox 
\epsfig{file=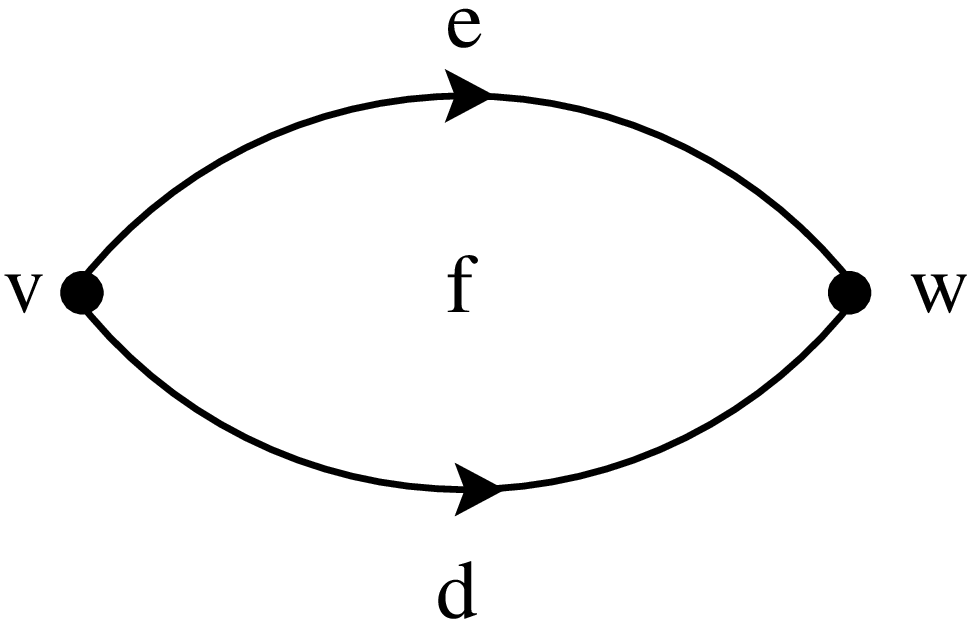,height=2.5cm}
\relabel{v}{$v$}
\relabel{w}{$w$}
\relabel{e}{$e$}
\relabel{f}{$f$}
\relabel{d}{$d$}
\endrelabelbox
${\quad 
\quad}$
\relabelbox 
\epsfig{file=2-cell-2.eps,height=2.5cm}
\relabel{v}{$v$}
\relabel{w}{$w$}
\relabel{e}{$\tilde{e}=d$}
\relabel{f}{$\tilde{f}$}
\relabel{d}{$\tilde{d}=e$}
\endrelabelbox}
\caption{\label{fig:vert-face-flip} Vertical inversion of a face }
\end{figure}
This leads to the following changes in the objects of $\cat{Conn}$, $(g,h)\mapsto (\tilde{g}, \tilde{h})$. 
The edges and the assignments to edges are unchanged, i.e. $g(e_i)\mapsto \tilde{g}(e_i)=g(e_i), \, \forall i$, 
and the assignments to faces are changed as follows:
\begin{equation}\label{eq:vertfaceflip}
\xybiglabels \vcenter{\xymatrix@M=0pt@=3pc{\ar@{-} [d] _{} \ar@{-} [r]^{g(e)} \ar@{} [dr]|{h(f)} & \ar@{-} [d]^{} \\
\ar@{-} [r]_{g(d)}  & }}
\quad \mapsto \quad
\xybiglabels \vcenter{\xymatrix@M=0pt@=3pc{\ar@{-} [d] _{} \ar@{-} [r]^{\tilde{g}(\tilde{e})} \ar@{} [dr]|{\tilde{h}(\tilde{f})} & \ar@{-} [d]^{} \\
\ar@{-} [r]_{\tilde{g}(\tilde{d})}  & }}
\quad = \quad
\xybiglabels \vcenter{\xymatrix@M=0pt@=3pc{\ar@{-} [d] _{} \ar@{-} [r]^{g(d)} \ar@{} [dr]|{h(f)^{-v}} & \ar@{-} [d]^{} \\
\ar@{-} [r]_{g(e)}  & }}
\end{equation}

\item[3)] change of face orientation by horizontal inversion

This change of orientation simultaneously affects the bigon structure and is given by 
the replacement of a 2-cell of $\cal{D}$ by a corresponding 2-cell of $\cal{D}'$ with 
the 0-source and 0-target exchanged, and the orientations of the 1-source and 1-target inverted:

\begin{figure}[htbp]
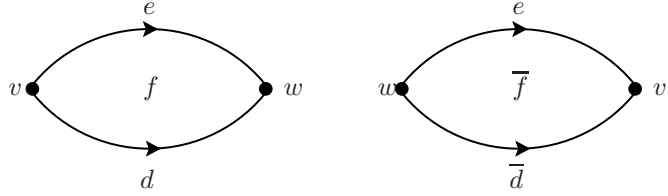
 
\centerline{\relabelbox 
\epsfig{file=2-cell-2.eps,height=2.5cm}
\relabel{v}{$v$}
\relabel{w}{$w$}
\relabel{e}{$e$}
\relabel{f}{$f$}
\relabel{d}{$d$}
\endrelabelbox
${\quad 
\quad}$
\relabelbox 
\epsfig{file=2-cell-2.eps,height=2.5cm}
\relabel{v}{$w$}
\relabel{w}{$v$}
\relabel{e}{$\overline{e}$}
\relabel{f}{$\overline{f}$}
\relabel{d}{$\overline{d}$}
\endrelabelbox}
\caption{\label{fig:hor-face-flip} Horizontal inversion of a face }
\end{figure}
This leads to the following changes in the objects of $\cat{Conn}$, $(g,h)\mapsto (\tilde{g}, \tilde{h})$. 
The edges and the assignments to edges are unchanged, as in 2), and the assignments to faces are changed as follows:
\begin{equation}\label{eq:horfaceflip}
\xybiglabels \vcenter{\xymatrix@M=0pt@=3pc{\ar@{-} [d] _{} \ar@{-} [r]^{g(e)} \ar@{} [dr]|{h(f)} & \ar@{-} [d]^{} \\
\ar@{-} [r]_{g(d)}  & }}
\quad \mapsto \quad
\xybiglabels \vcenter{\xymatrix@M=0pt@=3pc{\ar@{-} [d] _{} \ar@{-} [r]^{\tilde{g}(\overline{e})} \ar@{} [dr]|{\tilde{h}(\overline{f})} & \ar@{-} [d]^{} \\
\ar@{-} [r]_{\tilde{g}(\overline{d})}  & }}
\quad = \quad
\xybiglabels \vcenter{\xymatrix@M=0pt@=3pc{\ar@{-} [d] _{} \ar@{-} [r]^{g(e)^{-1}} \ar@{} [dr]|{h(f)^{-h}} & \ar@{-} [d]^{} \\
\ar@{-} [r]_{g(d)^{-1}}  & }}
\end{equation}

\item[4)] change of 0-source and 0-target

Consider the change in bigon structure on a face $f$ when we choose a different 0-source 
and 0-target, as in the following example, where $v,w$ are replaced by $v',w'$, the
1-source $e=e_1 \, e_2 \, e_3$ is replaced by $e'= e_2\,e_3\,\overline{d}_3$, and the 
1-target $d= d_1\,d_2\,d_3$ is replaced by  $d'= \overline{e}_1\,d_1\,d_2$, as shown in Figure
\ref{fig:orientedfaces}.

\begin{figure}[h]
  \begin{center}
    \includegraphics[height=3.5cm]{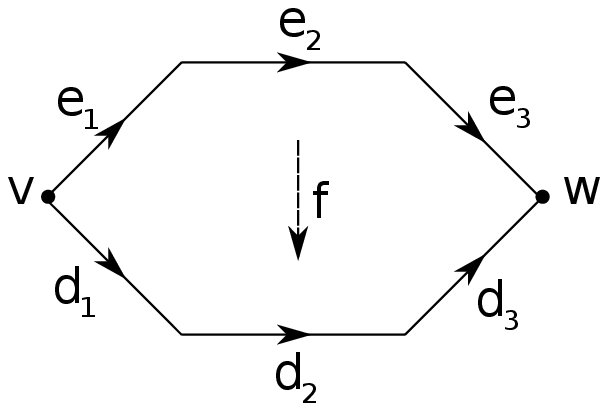}
\hskip 1.5cm
    \includegraphics[height=3.5cm]{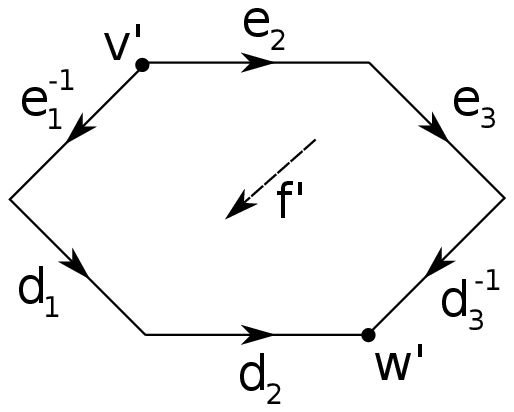}
    \caption{Two bigon structures on a face $f$}\label{fig:orientedfaces}
  \end{center}
\end{figure}
The face $f'$,
 regarded as a family of paths from the 1-source to the 1-target, is obtained from $f$ 
by ``whiskering'': 
\begin{equation*}
  f' = \overline{e}_1 \, f \, \overline{d}_3
\end{equation*}
This leads to the following changes in the objects of $\cat{Conn}$, $(g,h)\mapsto (\tilde{g}, \tilde{h})$.
The edges and the assignments to edges are unchanged, as in 2), and the assignments to faces are changed as follows:
\begin{equation}\label{eq:vwchange}
\xybiglabels \vcenter{\xymatrix@M=0pt@=3pc{\ar@{-} [d] _{} \ar@{-} [r]^{g(e)} \ar@{} [dr]|{h(f)} & \ar@{-} [d]^{} \\
\ar@{-} [r]_{g(d)}  & }}
\quad \mapsto \quad
 \xybiglabels \vcenter{\xymatrix@M=0pt@=3pc{\ar@{-} [d]  \ar@{-} [r]^{\tilde{g}(e')} \ar@{} [dr]|{\tilde{h}(f')} & \ar@{-} [d] \\ \ar@{-} [r]_{\tilde{g}(d')}  & }}
\quad = \quad
    \xybiglabels \vcenter{\xymatrix @=3pc @W=0pc @M=0pc { \ar@{-}[r] ^{g(e_1)^{-1}} \ar@{-}[d]_{} \ar@{}[dr]|{} & \ar@{-}[r] ^{g(e)} \ar@{-}[d]|{}
        \ar@{}[dr]|{h(f)} & \ar@{-}[r] ^{g(d_3)^{-1}} \ar@{-}[d]^{}  \ar@{}[dr]|{} & \ar@{-}[d]^{}
        \\ \ar@{-}[r] _{g(e_1)^{-1}} & \ar@{-}[r] _{g(d)} & \ar@{-}[r] _{g(d_3)^{-1}} &
      }}
\end{equation}

Formula (\ref{eq:vwchange}) is easily adapted to the general case, where a concatenation of edges $\nu$ from $v'$ to $v$ replaces $\overline{e}_1$ and a concatenation of edges $\omega$ from $w$ to $w'$ replaces $\overline{d}_3$.

\end{itemize}

\begin{theorem}
The correspondences (\ref{eq:edgeflip}), (\ref{eq:vertfaceflip}), (\ref{eq:horfaceflip}) and (\ref{eq:vwchange}) are functorial and yield isomorphisms of the categories 
$\cat{Conn}(M,\cal{D},\G)$ and $\cat{Conn}(M,\cal{D}',\G)$.
\label{thm:discr-change}
\end{theorem}

\begin{proof}
Each of the correspondences between objects extends to a functor $\cat{Conn}(M,\cal{D},\G)\rightarrow\cat{Conn}(M,\cal{D}',\G)$ via the following maps of morphisms 
$((g,h),\eta)\mapsto ((\tilde{g},\tilde{h}),\tilde{\eta})$:
\begin{itemize}
\item[1)] For edge $e_i$ we have:
\begin{equation}
\xybiglabels \vcenter{\xymatrix@M=0pt@=3pc{\ar@{-} [d] _{} \ar@{-} [r]^{g(e_i)} \ar@{} [dr]|{\eta(e_i)} & \ar@{-} [d]^{} \\
\ar@{-} [r]_{g'(e_i)}  & }}
\mapsto
\xybiglabels \vcenter{\xymatrix@M=0pt@=3pc{\ar@{-} [d] _{} \ar@{-} [r]^{\tilde{g}(\overline{e}_i)} \ar@{} [dr]|{\tilde{\eta}(\overline{e}_i)} & \ar@{-} [d]^{} \\
\ar@{-} [r]_{\tilde{g}'(\overline{e}_1)}  & }} = 
\xybiglabels \vcenter{\xymatrix@M=0pt@=3pc{\ar@{-} [d] _{} \ar@{-} [r]^{g(e_i)^{-1}} \ar@{} [dr]|{\eta(e_i)^{-h}} & \ar@{-} [d]^{} \\
\ar@{-} [r]_{g'(e_i)^{-1}}  & }}
\end{equation}
On faces, equation (\ref{eq:conn-morphisms})  of Def. \ref{def:catconn} is unchanged due to the orientation conventions 
given below  (\ref{eq:conn-morphisms}).
\item[2)] We have $\tilde{g}(e_i) =g(e_i)$ and we set $\tilde{\eta}(e_i) =\eta(e_i)$ for all $i$, so that the squares (\ref{eq:conn-morphisms-edge}) are unchanged. For the face of Figure \ref{fig:vert-face-flip}, we have:
\begin{equation}\label{eq:ver-flip-mor}
\xybiglabels \vcenter{\xymatrix@M=0pt@=3pc{\ar@{-} [d] _{} \ar@{-} [r]^{\tilde{g}(\tilde{e})} \ar@{} [dr]|{\tilde{\eta}(\tilde{e})} & \ar@{-} [d]^{} \\
\ar@{-} [r]_{\tilde{g}'(\tilde{e})}  & }} =
\xybiglabels \vcenter{\xymatrix@M=0pt@=3pc{\ar@{-} [d] _{} \ar@{-} [r]^{g(d)} \ar@{} [dr]|{\eta(d)} & \ar@{-} [d]^{} \\
\ar@{-} [r]_{g'(d)}  & }},
\qquad
\xybiglabels \vcenter{\xymatrix@M=0pt@=3pc{\ar@{-} [d] _{} \ar@{-} [r]^{\tilde{g}(\tilde{d})} \ar@{} [dr]|{\tilde{\eta}(\tilde{d})} & \ar@{-} [d]^{} \\
\ar@{-} [r]_{\tilde{g}'(\tilde{d})}  & }} =
\xybiglabels \vcenter{\xymatrix@M=0pt@=3pc{\ar@{-} [d] _{} \ar@{-} [r]^{g(e)} \ar@{} [dr]|{\eta(e)} & \ar@{-} [d]^{} \\
\ar@{-} [r]_{g'(e)}  & }}
\end{equation}
It is a simple exercise to show that (\ref{eq:conn-morphisms}) for $((g,h),\eta)$ implies (\ref{eq:conn-morphisms}) for $((\tilde{g},\tilde{h}),\tilde{\eta})$, using 
(\ref{eq:vertfaceflip}) and (\ref{eq:ver-flip-mor}), together with property (\ref{eq:ver-inverse-ver-prod}) of the vertical inverse.

\item[3)] As in the previous case, we have $\tilde{g}(e_i) =g(e_i)$ and we set $\tilde{\eta}(e_i) =\eta(e_i)$ for all $i$, so that the squares (\ref{eq:conn-morphisms-edge}) are unchanged.
For the face of Figure \ref{fig:hor-face-flip}, we have:
\begin{equation}\label{eq:hor-flip-mor}
\xybiglabels \vcenter{\xymatrix@M=0pt@=3pc{\ar@{-} [d] _{} \ar@{-} [r]^{\tilde{g}(\overline{e})} \ar@{} [dr]|{\tilde{\eta}(\overline{e})} & \ar@{-} [d]^{} \\
\ar@{-} [r]_{\tilde{g}'(\overline{e})}  & }} =
\xybiglabels \vcenter{\xymatrix@M=0pt@=3pc{\ar@{-} [d] _{} \ar@{-} [r]^{g(e)^{-1}} \ar@{} [dr]|{\eta(e)^{-h}} & \ar@{-} [d]^{} \\
\ar@{-} [r]_{g'(e)^{-1}}  & }},
\qquad
\xybiglabels \vcenter{\xymatrix@M=0pt@=3pc{\ar@{-} [d] _{} \ar@{-} [r]^{\tilde{g}(\overline{d})} \ar@{} [dr]|{\tilde{\eta}(\overline{d})} & \ar@{-} [d]^{} \\
\ar@{-} [r]_{\tilde{g}'(\overline{d})}  & }} =
\xybiglabels \vcenter{\xymatrix@M=0pt@=3pc{\ar@{-} [d] _{} \ar@{-} [r]^{g(d)^{-1}} \ar@{} [dr]|{\eta(d)^{-h}} & \ar@{-} [d]^{} \\
\ar@{-} [r]_{g'(d)^{-1}}  & }}
\end{equation}
Equation (\ref{eq:conn-morphisms}) for $((\tilde{g},\tilde{h}),\tilde{\eta})$ follows directly from  (\ref{eq:conn-morphisms}) for $((g,h),\eta)$ by horizontal inversion, using property 
(\ref{eq:ver-inverse-ver-prod}) of the horizontal inverse.

\item[4)] As in the two previous cases, we have $\tilde{g}(e_i) =g(e_i)$ and we set $\tilde{\eta}(e_i) =\eta(e_i)$ for all $i$, so that the squares (\ref{eq:conn-morphisms-edge}) are unchanged.
For the face of Figure \ref{fig:orientedfaces} we define $\tilde{\eta}$ by
\begin{equation}\label{eq:vwchange-mor}
 \xybiglabels \vcenter{\xymatrix@M=0pt@=3pc{\ar@{-} [d]  \ar@{-} [r]^{\tilde{g}(e')} \ar@{} [dr]|{\tilde{\eta}(e')} & \ar@{-} [d] \\ \ar@{-} [r]_{\tilde{g}'(e')}  & }}
\quad = \quad
    \xybiglabels \vcenter{\xymatrix @=3pc @W=0pc @M=0pc { \ar@{-}[r] ^{g(e_1)^{-1}} \ar@{-}[d]_{} \ar@{}[dr]|{} & \ar@{-}[r] ^{g(e)} \ar@{-}[d]|{}
        \ar@{}[dr]|{\eta(e)} & \ar@{-}[r] ^{g(d_3)^{-1}} \ar@{-}[d]^{}  \ar@{}[dr]|{} & \ar@{-}[d]^{}
        \\ \ar@{-}[r] _{g(e_1)^{-1}} & \ar@{-}[r] _{g'(e)} & \ar@{-}[r] _{g(d_3)^{-1}} &
      }}
\end{equation}
and an analogous equation with $e',\, e$ replaced by $d',\, d$. Equation (\ref{eq:conn-morphisms}) for $((\tilde{g},\tilde{h}),\tilde{\eta})$ then follows directly from  (\ref{eq:conn-morphisms}) for $((g,h),\eta)$ by composing horizontally on the left and right with identity squares for $g(e_1)^{-1}$ and $g(d_3)^{-1}$ and using (\ref{eq:vwchange}) and (\ref{eq:vwchange-mor}).

\end{itemize}

The functors we have obtained preserve identities and composition, as can be easily verified. With regards to composition and the first correspondence, this is based on
$$
\widetilde{(\eta'\eta)}(\overline{e}_i)= (\eta'\eta)^{-h}(e_i) = \eta'^{-h}(e_i) \eta^{-h}(e_i) = \tilde{\eta}'(\overline{e}_i) \tilde{\eta}(\overline{e}_i)
$$
where the square notation is understood. For the other three correspondences $\tilde{g}=g$ and $\tilde{\eta}=\eta$ on edges, so composition is obviously preserved. Since the maps, both at the object and morphism level, are easily seen to be invertible, the correspondences yield isomorphisms of the categories 
$\cat{Conn}(M,\cal{D},\G)$ and $\cat{Conn}(M,\cal{D}',\G)$.

\end{proof}

\begin{remark}
Theorem \ref{thm:discr-change} is in the same spirit as
 \cite[Def. 54, Lemma 55]{bullivant-et-al} concerning a change of the single basepoint in the boundary of a 2-cell.

\end{remark}

The preceding theorem tells about how the groupoid $\cat{Conn}$
changes when we change the orientations and bigon structure of the
 cell structure. This shows that, up to isomorphism, $\cat{Conn}$ does
not depend on the choice of $(O,B)$, but only on the cells themselves and the
associated attaching maps from the CW-structure. Since in general one
is more interested in bare manifolds than ones equipped with a
CW-structure, it is natural to ask whether analogous results hold when
one changes this part of $\cal{D}$. This leads us toward connections
in the usual smooth sense.

To address this question fully would require either or both of two
bodies of theory which are more than we wish to engage with here:
namely, a fuller category-theoretic treatment of double categories,
and the analytic techniques used to handle the infinite-dimensional
manifolds coming from spaces of functions and $p$-forms.
 Thus, we simply present two conjectures which suggest a
possible line of inquiry.

First, consider the case where we restrict our attention to flat
connections and the moduli spaces involved are
finite-dimensional. Denote by  $\cat{Conn}_0$ the category of
\textit{flat} $\G$-connections on $(M,\cal{D})$. We hope to recover an
analog of a result which holds for 1-groups: different discretizations
yield transformation groupoids which are Morita
equivalent.

Unfortunately, the notion of Morita
equivalence, i.e. equivalence at the level of the representation categories,
 has not yet been sufficiently developed for double categories, as far as we know, 
and the issue
of which higher category of double categories one works with could affect the 
representation category. Thus we are not in a position 
to state, much less prove, an analogous theorem.
However, we conjecture that, for a suitable
notion of Morita equivalence of double categories, the following will
be true:

\begin{conjecture}\label{conj:morita}
  If ${\cal D}_1$ and ${\cal D}_2$ are two
  different discretizations of a manifold $M$, then the double groupoids
  $\cat{Conn}_0 \wquot \cat{Gauge}$ for 
  $(M,{\cal D}_1, \G)$ and   $(M,{\cal D}_2, \G)$ are
  Morita equivalent.
\end{conjecture}

For 1-groupoids, Morita equivalent groupoids describe equivalent
physical situations. This is a key idea, for example, behind
symplectic reduction. Providing that a suitable definition of Morita
equivalence for double groupoids has the same property, we would then
be able to say that, for flat connections, the choice of $\cal{D}$ is
purely a convenience. Indeed, we propose this physical equivalence as
a useful criterium for a suitable definition of Morita equivalence in
this context.

If we are not considering flat connections, then of course we should
not expect any such result. Rather, we are then treating $\cal{D}$ as
a ``probe'' of a connection, which gives a finite approximation to the
continuum theory by taking holonomies along particular edges and
faces. At best, we may hope that there is a double groupoid 
$\cat{Conn} \wquot \cat{Gauge}(M,\G)$ that is a
limit over all
discretizations as the probes are taken to be increasingly finer. This at least
potentially makes sense, since there is a partial order relation on
all CW-structures by refinement.  In particular, 
for 2-dimensional manifolds $(V,E,F)$ is a refinement of
$(V',E',F')$ if $V'$ is a subset of $V$, every edge in $E$ lies within
an edge of $E'$, and every face in $F$ lies within a single face of
$F'$. Clearly, knowing holonomies of a connection on $(M,{\cal D'})$ is
sufficient to determine them on $(M,{\cal D})$, and similarly for
gauge transformations. So one can think of approaching the continuum
limit through successive refinements of discretization. This leads us
to the following conjecture:

\begin{conjecture}\label{conj:continuumlimit}
 $\cat{Conn} \wquot \cat{Gauge}(M,\G)$ is the 
inductive limit of $\cat{Conn} \wquot \cat{Gauge}(M,{\cal D},\G)$
   over all discretizations ${\cal D}$ of $M$.
\end{conjecture}

This conjecture might need to be improved in the light of analytic
considerations about the presumably infinite-dimensional spaces of
objects, morphisms, and squares in $\cat{Conn} \wquot \cat{Gauge}(M,\G)$, and how they arise as
limits of finite-dimensional spaces, but does suggest the form of the
hoped-for result.

In the case of flat connections, and in combination with Conjecture
\ref{conj:morita}, it would imply that to find $\cat{Conn}_0 \wquot \cat{Gauge}(M,\G)$ up to
Morita equivalence, we need only find $\cat{Conn}_0 \wquot \cat{Gauge}(M,{\cal D},\G)$ for
\textit{any} discretization ${\cal D}$ of $M$.

\section{The Transformation Double Groupoid for Higher Connections}
\label{sec:transdoublegpd}

It is well known that the action of a group $G$ on a set $X$ can be described by a (transformation) groupoid $X\wquot G$, with objects $X$ and morphisms of the form 
$x \stackrel{(g,x)}{\longrightarrow} g.x$. In \cite{morton-picken-i} we showed that an analogous situation occurs for the action of a 2-group $\G$ on a category $\cat{C}$, namely this ``categorified action'' can be described by a transformation {\em double} category, $\cat{C}\wquot \G$, which becomes a transformation double groupoid when the category $\cat{C}$ is a groupoid. Below we review the construction of \cite{morton-picken-i} and present the transformation double groupoid $\cat{Conn}\wquot \cat{Gauge}$ which arises in our case.

\subsection{The transformation double category  \texorpdfstring{$\cat{C}\wquot \G$}{} }

Given an action of a  categorical group $\G$ on a category $\C$, as in Def. \ref{def:2grp_action-Phi-hat}, we showed in \cite{morton-picken-i} that one can define a double category $\cat{C}\wquot \G$, with objects being the objects of $\C$, horizontal morphisms being the morphisms of $\C$, vertical morphisms being of the form $x \stackrel{(\gamma,x)}{\longrightarrow} \gamma\act x$, where $x\in {\rm Ob} \, X$, $\gamma\in{\rm Ob}\, \G$, and squares being ${\rm Mor}\, \G \times {\rm Mor}\, \C$, with horizontal and vertical sources and targets given by:
 \begin{equation}
    \xymatrix@C=7.5pc@R=4pc {
      x \ar[r]^{f} \ar[d]_{(\gamma, x)} \drtwocell<\omit>{\omit *+[F]{(\gamma,\chi),f}}  &  y \ar[d]^{(\partial(\chi)\gamma,y)}  \\
      {\gamma \act x} \ar[r]_{(\gamma,\chi)\act f}  & {(\partial(\chi)\gamma)\act y}
    }
  \label{eq:squaredef}   
  \end{equation}
Here $\gamma\act x$ denotes $\hat{\Phi}(\gamma,x)$ and $(\gamma,\chi)\act f$ denotes $\hat{\Phi}((\gamma,\chi), f)$, as in Def. \ref{def:2grp_action-Phi-hat}.
In \cite{morton-picken-i} we prove that this is indeed a double category with suitable horizontal and vertical composition of the respective morphisms, and horizontal and vertical composition of squares. The double category $\C \wquot G$ captures in a single structure several different group actions that are at work simultaneously, namely the action of the objects of $\G$ on the objects of $\C$, the action of the objects of $\G$ on the morphisms of $\C$, and finally the action of the morphisms of $\G$ on the morphisms of $\C$. See \cite[Def. 3.5]{morton-picken-i} for a detailed exposition. The vertical morphisms of $\C \wquot G$ are invertible, and when the horizontal morphisms are also invertible, i.e. $\C$ is a groupoid,  $\C \wquot G$ becomes a double groupoid.

\subsection{The Transformation Double Groupoid \texorpdfstring{$\cat{Conn} \wquot \cat{Gauge}$}{} }

We now give a detailed definition of the transformation double groupoid $\cat{Conn} \wquot \cat{Gauge}$ that arises in our specific case.

\begin{definition}
\label{def:Conn-over-Gauge}
Given an action of $\cat{Gauge}$ on $\cat{Conn}$, as defined in Def. \ref{def:Phi-hat-for-Conn+Gauge}, the transformation double groupoid $\cat{Conn} \wquot \cat{Gauge}$ is given by:

\begin{itemize}
\item \textbf{Objects} are the objects $(g,h)$ of  $\cat{Conn}$ 

\item \textbf{Horizontal morphisms} are the morphisms $((g,h),\eta)$ of  $\cat{Conn}$, with source maps, target maps and horizontal composition defined as in  $\cat{Conn}$ 

\item \textbf{Vertical morphisms}  are the set of pairs $(\gamma, (g,h))$, where $\gamma:V\ra G$ is an object of  $\cat{Gauge}$ and $(g,h)$ is an object of $\cat{Conn}$. The source of $(\gamma, (g,h))$ is $(g,h)$ and the target is $(\gamma.g, \gamma.h)$. Composition of vertical morphisms is defined by:
$$
(\tilde{\gamma}, (\gamma.g, \gamma.h))\circ (\gamma, (g,h)) = (\tilde{\gamma}\gamma, (g,h))
$$

\item \textbf{Squares} are the set of pairs of morphisms of $\cat{Gauge}$ and $\cat{Conn}$, and are denoted $$
\squaremor{(\gamma, \chi), ((g,h),\eta)}\, .
$$ 
Horizontal and vertical sources and targets are given as follows:
  \begin{equation}
\xybiglabels{    \xymatrix@C=7.5pc@R=4pc {
      (g,h) \ar[r]^{((g,h),\eta)} \ar[d]_{(\gamma, (g,h))} \drtwocell<\omit>{\omit *+[F]{(\gamma,\chi),((g,h),\eta)}}  &  (g',h') \ar[d]^{(\gamma',(g',h')) }  \\
      {(\gamma.g, \gamma.h) } \ar[r]_{((\gamma.g, \gamma.h), (\gamma,\chi).\eta)}  & {(\gamma'.g',\gamma'.h')}
    }}
  \label{eq:square-sourcetarget}   
  \end{equation}

\item \textbf{Composition} Horizontal and vertical composition of squares are given by:
  \begin{equation}\label{eq:squarehorizcomp}
    \squaremor{ (\gamma', \chi'),((g',h'),\eta') } 
    \, \circ_h \,  \squaremor{(\gamma, \chi),((g,h),\eta) } 
    \, = \, \squaremor{ (\gamma, \chi'\chi),((g,h),\eta'\eta) }
  \end{equation}
 and:
  \begin{equation}\label{eq:squarevertcomp}
  \squaremor{ (\tilde{\gamma}, \tilde{\chi}), ((\gamma.g,\gamma.h), (\gamma, \chi).\eta) } 
   \,  \circ_v \,     \squaremor{ (\gamma, \chi),((g,h),\eta) }
 \,   = \, \squaremor{ (\tilde{\gamma}\gamma, \tilde{\chi} (\tilde{\gamma}\rhd\chi)), ((g,h),\eta)}
  \end{equation}
\end{itemize}

\end{definition}

\begin{remark}
  As noted in the introduction, we remark here (because the terms will
  appear throughout our forthcoming paper \cite{morton-picken-iii}) that
  we also refer to the horizontal morphisms of this double groupoid as
  \textit{costrict gauge transformations}, and the vertical morphisms
  as \textit{strict gauge transformations}. The squares, in this
  usage, will be called \textit{gauge modifications} between such
  transformations.
\end{remark}

\section{Geometrical Examples}\label{sec:geom-examples}

The previous approach was set up to handle any finite number of cells in a tendentially local description of connections, but for simple manifolds it also allows an efficient global description of the action of $\cat{Gauge}$ on $\cat{Conn}$, employing a small number of cells.

The manifolds in our examples will all be of dimension less than 3, so that the connections will be automatically flat, since the curvature 3-form vanishes.

In each of the following examples we will highlight special features of the action.

\subsection{The example of the circle}

The circle can be endowed with a cell decomposition consisting of a single 0-cell $v$ and a single 1-cell $e$ - see Figure \ref{fig:S1-cell}

\begin{figure}[h]
\begin{center}
\includegraphics[height=1cm]{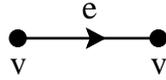}
\end{center}
\caption{Cell decomposition for the circle $S^1$}
\label{fig:S1-cell}
\end{figure}

Since there are no 2-cells, setting $g(e)=g$ and $\eta(e)=\eta$, the objects $\{(g(e),1) \}$ of $\cat{Conn}$ may be identified with $G$ and the morphisms ${\rm Mor}(\cat{Conn})$ with $ G\times H   $. Thus
$$
\cat{Conn}(S^1) \cong \G.
$$
Likewise setting $\gamma(v)=\gamma$ and $\chi(v)=\chi$, the 2-group $\cat{Gauge}$ may be identified with $\G$, since we have a single 0-cell, i.e.
$$
\cat{Gauge}(S^1) \cong \G.
$$

The action of $\cat{Gauge}$ on $\cat{Conn}$ is then the adjoint action of $\G$ on itself, as described in \cite{morton-picken-i}, given by particularizing (\ref{eq:Phimor}):

  \begin{equation}
\label{eq:S1-action}
    \xybiglabels \vcenter{\xymatrix@M=0pt@=3pc@C=5pc{\ar@{-} [d]  \ar@{-} [r]^{\gamma.g} \ar@{} [dr]|{(\gamma,\chi).\eta} & \ar@{-} [d] \\ \ar@{-} [r]_{(\gamma.g)'}  & }}
    \,  = \,
    \xybiglabels \vcenter{\xymatrix @=3pc @W=0pc @M=0pc { \ar@{-}[r] ^{\gamma} \ar@{-}[d]_{} \ar@{}[dr]|{\chi} & \ar@{-}[r] ^{g} \ar@{-}[d]|{}
        \ar@{}[dr]|{\eta} & \ar@{-}[r] ^{\gamma^{-1}} \ar@{-}[d]^{}  \ar@{}[dr]|{\chi^{-h}} & \ar@{-}[d]^{}
        \\ \ar@{-}[r] _{\gamma'} & \ar@{-}[r] _{g'} & \ar@{-}[r] _{\gamma'^{-1}} &
      }}.
  \end{equation}
  
  Thus for the case of $S^1$, the transformation double groupoid arising from this action is given by:
$$
\cat{Conn} \wquot \cat{Gauge} (S^1) \cong \G \wquot \G
$$
 where the adjoint action is understood, in complete  analogy with the gauge groupoid $G \wquot G$ for ordinary $G$-gauge theory on the circle. The squares of this transformation double groupoid are:
  \begin{equation}
\xybiglabels{    \xymatrix@C=7.5pc@R=4pc {
      g \ar[r]^{(g,\eta)} \ar[d]_{(\gamma, g)} \drtwocell<\omit>{\omit *+[F]{(\gamma,\chi),(g,\eta)}}  &  g' \ar[d]^{(\gamma',g') }  \\
      {\gamma.g } \ar[r]_{(\gamma.g, (\gamma,\chi).\eta)}  & {\gamma'.g'}
    }}
  \label{eq:square-S1}   
  \end{equation}
We refer the reader to \cite{morton-picken-i} for a detailed description of the transformation double groupoid $\G\wquot \G$. 

Finally a suggestive 3-dimensional perspective on the action is given by taking a commuting 3-cube and identifying two opposite faces as in Figure \ref{fig:3D-circle}. Note that to simplify this figure we have abbreviated:
\begin{align}
  \eta_1 & = (\gamma,\chi).\eta \nonumber\\
  g_1 & = \gamma . g \nonumber\\
  g'_1 & = \gamma' . g' \nonumber
\end{align}

\begin{figure}[h]
  \begin{center}
    \includegraphics{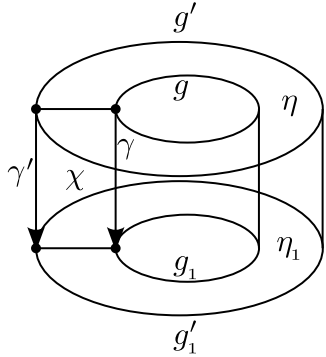}
  \end{center}
  \caption{3D perspective for $S^1$}\label{fig:3D-circle}
\end{figure}

\subsection{The example of the sphere }
\label{subsec:S2-exp}
The sphere $S^2$ can be realized as a single 2-cell $f$ with the bigon structure as in Figure  \ref{fig:S2-cell}, where the 1-cell edge $d$ is identified with $e$. 

\begin{figure}[h]
\begin{center}
\includegraphics[height=3cm]{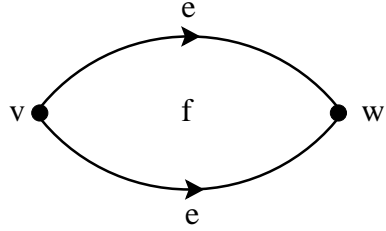}
\end{center}
\caption{Cell decomposition for the sphere $S^2$}
\label{fig:S2-cell}
\end{figure}

We start by describing $\cat{Conn}(S^2)$. Setting $h(f)=h$, $g(e)=g$ one has $\partial(h)=gg^{-1}=1$. Hence the objects of $\cat{Conn}(S^2)$ are 
$\{(g,h)\,|\, \partial(h)=1  \}=G\times {\rm ker}(\partial)$. Setting $\eta(e)=\eta$, the morphisms of $\cat{Conn}(S^2)$ are 
$\{((g,h),\eta) \,|\, (g,h)\in G\times  {\rm ker}(\partial), h\in H \} $ and the target of $((g,h),\eta)$ is $(g',h')$, where $g'=\partial(\eta)g$, from (\ref{eq:conn-morphisms-edge}), and $h'$ is given, from (\ref{eq:conn-morphisms}), by vertical conjugation of $h$ by $\eta$ (see Remark \ref{rem:h'-formula}):
\begin{equation*}
\xybiglabels \vcenter{\xymatrix@M=0pt@=3pc{\ar@{-} [d] _{} \ar@{-} [r]^{g'} \ar@{} [dr]|{h'} & \ar@{-} [d]^{} \\
\ar@{-} [r]_{g'}  & }} = 
\xybiglabels \vcenter{\xymatrix@M=0pt@=3pc{\ar@{-} [d] _{} \ar@{-} [r]^{g'} \ar@{} [dr]|{\eta^{-v}} & \ar@{-} [d]^{} \\
\ar@{-} [r] |{g} \ar@{-} [d]_{} \ar@{}[dr] |{h} & \ar@{-} [d]^{} \\
\ar@{-} [r] |{g} \ar@{-} [d]_{} \ar@{}[dr] |{\eta} & \ar@{-} [d]^{} \\
\ar@{-} [r]_{g'}& }} 
\end{equation*}

Since there are two 0-cells, $v$ and $w$, we have 
$$
\cat{Gauge}(S^2) \cong \G\times \G
$$
Setting $\gamma(v)=\gamma_1$, $\gamma(w)=\gamma_2$, 
$\chi(v)=\chi_1$, $\chi(w)=\chi_2$, the action of $(\gamma_1,\gamma_2)$ on $(g,h)$ is described by (\ref{eq:Phi-ob-face}):
\begin{equation}
\xybiglabels \vcenter{\xymatrix@M=0pt@=3pc@C=4pc{\ar@{-} [d] \ar@{-} [r]^{\gamma. g} \ar@{} [dr]|{\gamma.h} & \ar@{-} [d] \\
\ar@{-} [r]_{\gamma. g}  & }}
\, = \,
\xybiglabels \vcenter{\xymatrix @=3pc @W=0pc @M=0pc { \ar@{-}[r] ^{\gamma_1} \ar@{-}[d]_{} \ar@{}[dr]|{} & \ar@{-}[r] ^{g} \ar@{-}[d]|{}
\ar@{}[dr]|{h} & \ar@{-}[r] ^{\gamma_2^{-1}} \ar@{-}[d]^{}  \ar@{}[dr]|{} & \ar@{-}[d]^{}
\\ \ar@{-}[r] _{\gamma_1} & \ar@{-}[r] _{g} & \ar@{-}[r] _{\gamma_2^{-1}} &
}} 
\, = \,
\xybiglabels \vcenter{\xymatrix@M=0pt@=3pc@C=4pc{\ar@{-} [d] \ar@{-} [r]^{\gamma_1 g \gamma_2^{-1}} \ar@{} [dr]|{\gamma_1\rhd h} & \ar@{-} [d] \\
\ar@{-} [r]_{\gamma_1 g \gamma_2^{-1}}  & }}
\label{eq:S2gauge}
\end{equation}
and the action of $(\gamma,\chi)=((\gamma_1,\gamma_2),(\chi_1,\chi_2))$ on $(g,\eta)$ is described by (\ref{eq:Phimor}):
$$
    \xybiglabels \vcenter{\xymatrix@M=0pt@=3pc@C=5pc{\ar@{-} [d]  \ar@{-} [r]^{\gamma.g} \ar@{} [dr]|{(\gamma,\chi).\eta} & \ar@{-} [d] \\ \ar@{-} [r]_{\gamma'.g'}  & }}
    \,  = \,
    \xybiglabels \vcenter{\xymatrix @=3pc @W=0pc @M=0pc { \ar@{-}[r] ^{\gamma_1} \ar@{-}[d]_{} \ar@{}[dr]|{\chi_1} & \ar@{-}[r] ^{g} \ar@{-}[d]|{}
        \ar@{}[dr]|{\eta} & \ar@{-}[r] ^{\gamma_2^{-1}} \ar@{-}[d]^{}  \ar@{}[dr]|{\chi_2^{-h}} & \ar@{-}[d]^{}
        \\ \ar@{-}[r] _{\gamma_1'} & \ar@{-}[r] _{g'} & \ar@{-}[r] _{\gamma_2'^{-1}} &
      }}.
$$
Thus $\cat{Gauge}(S^2)$ acts on $\cat{Conn}(S^2)$ via two independent $\G$ actions from the left and from the right.

\begin{remark}
In \cite[Examples 74, 75]{bullivant-et-al} the holonomy along $S^2$ is approached using two different cell structures.
The first has a single vertex and a single 2-cell attached to it, and the second is a subdivision of
our cell decomposition of Fig. \ref{fig:S2-cell}, using four 1-cells from $v$ to $w$ to divide $f$ into four faces.
In both cases the holonomy, corresponding to our $h(f)$, is likewise given in \cite{bullivant-et-al} by an element of ${\rm ker}(\partial)$.
Under full gauge transformations, corresponding to our action (\ref{eq:S2gauge}), the result \cite[Thm. 97]{bullivant-et-al}
is the same as ours, namely the holonomy $h(f)$ is transformed into $\gamma(v)\rhd h(f)$, where $v$ is a vertex on the 2-sphere surface.

\end{remark}

\subsection{The example of the torus}
\label{subsec:T2-exp}

The torus $T^2$ can be realized with a single 0-cell $v$, two 1-cells $e_1$ and $e_2$, and a single 2-cell $f$ with bigon structure as depicted in Figure \ref{fig:T2-cell}. 
Note that  this is an example of a 2-cell with bigon structure, as in Figure \ref{fig:bigon-struc}, where $e$ and $d$ are both concatenations of more than one 1-cell.

\begin{figure}[h]
\begin{center}
\includegraphics[height=2.5cm]{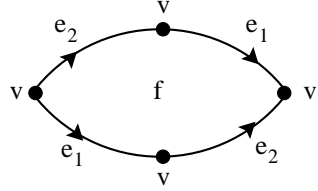}
\end{center}
\caption{Cell decomposition for the torus $T^2$}
\label{fig:T2-cell}
\end{figure}

We start by describing $\cat{Conn}(T^2)$. Setting $g(e_1)=g_1$, $g(e_2)=g_2$, and $h(f)=h$, we have $\partial(h)=g_1g_2g_1^{-1}g_2^{-1}$, and hence the objects of $\cat{Conn}(T^2)$ are of the form:
$$
((g_1,g_2),h) \in G^2\times \partial^{-1}([G,G]),
$$
where $[G,G]$ denotes the commutator subgroup of $G$. Likewise setting $\eta(e_1)=\eta_1$, $\eta(e_2)=\eta_2$, the morphisms of $\cat{Conn}(T^2)$ 
are given by: 
$$
(((g_1,g_2),h), (\eta_1, \eta_2)) \in G^2\times \partial^{-1}([G,G]) \times H^2.
$$
The target  of a morphism $(((g_1,g_2),h), (\eta_1, \eta_2))$ is $((g_1',g_2'),h')$, where $g_i'=\partial(\eta_i)g_i$, and $h'$ is given by

$$
\xybiglabels \vcenter{\xymatrix@M=0pt@=3pc{\ar@{-} [d] _{} \ar@{-} [r]^{g_2'g_1'} \ar@{} [dr]|{h'} & \ar@{-} [d]^{} \\
\ar@{-} [r]_{g_1'g_2'}  & }}
=
\xybiglabels \vcenter{\xymatrix @=3pc @W=0pc @M=0pc @C=3pc { 
\ar@{-}[r] ^{g_2'} \ar@{-}[d]_{} \ar@{}[dr]|{\eta_2^{-v}} & \ar@{-}[r] ^{g_1'} \ar@{-}[d]|{} \ar@{}[dr]|{\eta_1^{-v}} & \ar@{-}[d]^{} \\ 
\ar@{-}[r] |{g_2} \ar@{-}[d]_{} \ar@{}[dr]|{} & \ar@{-}[r] |{g_1} \ar@{}[d]|{h} \ar@{}[dr]|{}   & \ar@{-}[d]^{} \\ 
\ar@{-}[r] |{g_1} \ar@{-}[d]_{} \ar@{}[dr]|{\eta_1} & \ar@{-}[r] |{g_2} \ar@{-}[d]|{} \ar@{}[dr]|{\eta_2} & \ar@{-}[d]^{} \\ 
\ar@{-}[r] _{g_1'} & \ar@{-}[r] _{g_2'}  &
}}
$$
where, on the right hand side, the horizontal compositions are performed first.

Since there is a single 0-cell $v$. we have
$$
\cat{Gauge}(T^2)\cong \G
$$
Setting $\gamma(v)=\gamma$, $\chi(v)=\chi$, the action of $\gamma$ on $((g_1,g_2),h)$ is given by $\gamma.g_i= \gamma g_i \gamma^{-1}$, for $i=1,2$, and
$$
\xybiglabels \vcenter{\xymatrix@M=0pt@=3pc@C=5pc{\ar@{-} [d] \ar@{-} [r]^{(\gamma. g_2)(\gamma. g_1)} \ar@{} [dr]|{\gamma.h} & \ar@{-} [d] \\
\ar@{-} [r]_{(\gamma. g_1)(\gamma. g_2)}  & }}
\, = \,
\xybiglabels \vcenter{\xymatrix @=3pc @W=0pc @M=0pc { \ar@{-}[r] ^{\gamma} \ar@{-}[d]_{} \ar@{}[dr]|{} & \ar@{-}[r] ^{g_2g_1} \ar@{-}[d]|{}
\ar@{}[dr]|{h} & \ar@{-}[r] ^{\gamma^{-1}} \ar@{-}[d]^{}  \ar@{}[dr]|{} & \ar@{-}[d]^{}
\\ \ar@{-}[r] _{\gamma} & \ar@{-}[r] _{g_1g_2} & \ar@{-}[r] _{\gamma^{-1}} &
}} 
$$

The action of $(\gamma,\chi)$ on $(((g_1,g_2),h),\eta)$ is described by:
$$
\xybiglabels \vcenter{\xymatrix@M=0pt@=3pc@C=5pc{\ar@{-} [d] \ar@{-} [r]^{\gamma. g_i} \ar@{} [dr]|{(\gamma,\chi).\eta_i} & \ar@{-} [d] \\
\ar@{-} [r]_{\gamma'. g_i'}  & }}
\, = \,
\xybiglabels \vcenter{\xymatrix @=3pc @W=0pc @M=0pc { \ar@{-}[r] ^{\gamma} \ar@{-}[d]_{} \ar@{}[dr]|{\chi} & \ar@{-}[r] ^{g_i} \ar@{-}[d]|{}
\ar@{}[dr]|{\eta_i} & \ar@{-}[r] ^{\gamma^{-1}} \ar@{-}[d]^{}  \ar@{}[dr]|{\chi^{-h}} & \ar@{-}[d]^{}
\\ \ar@{-}[r] _{\gamma'} & \ar@{-}[r] _{g_i'} & \ar@{-}[r] _{\gamma'^{-1}} &
}} 
$$
for $i=1,2$.

The whole action can be captured in a single 4D diagram, depicted in Figure \ref{fig:2gauge-torus}, consisting of eight commuting 3-cubes (the inner 3-cube, its six adjacent 3-cubes and the outer 3-cube). 
In this figure, squares of a more general type appear, with possibly non-trivial labels on the side edges, as opposed to the squares (\ref{eq:square}).
Focussing on the uppermost 3-cube, which describes the relation between the source and the target of the morphism $(((g_1,g_2),h), (\eta_1, \eta_2))$, we see that the target can be viewed as coming from the source by simultaneous horizontal and vertical conjugation by $\eta_1$ and $\eta_2$ - compare with \cite[Thm. 5.17]{martinspickenii}.

\begin{figure}[h]
  \begin{center}
    \includegraphics[height=6.9cm]{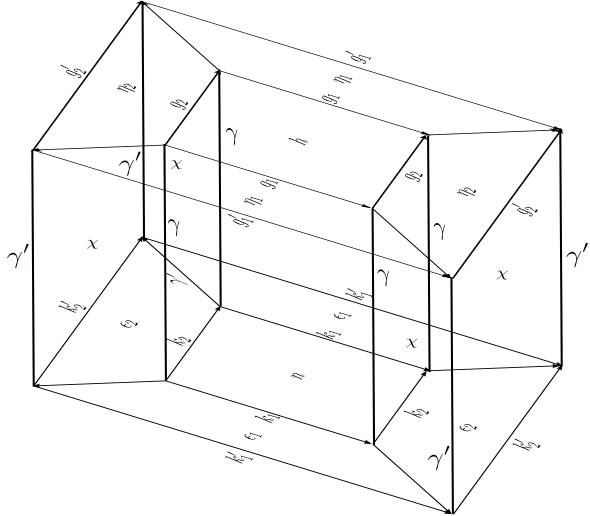}
  \end{center}
  \caption{4D perspective for the torus}\label{fig:2gauge-torus}
\end{figure}

\newpage

Note that in this figure, we have used the following shorthand
notation to maintain readability:
\begin{align}
  k_i & = \gamma g_i \gamma^{-1} \nonumber \\
  g'_i & = \partial (\eta_i) g_i \nonumber \\
  \epsilon_i & = \chi \eta_i \chi^{-h}\nonumber  \\
  k'_i & = \partial (\epsilon_i) k_i \nonumber
\end{align}


\section*{Acknowledgements}

We are grateful to João Faria Martins and Björn Gohla for discussions.
This work was supported in part by the Deutsche Forschungsgemeinschaft (Germany) 
through the Research Training Group 1670 (Mathematics Inspired by String Theory and 
Quantum Field Theory) of the University of Hamburg,  and by the 
Fundação para a Ciência e a Tecnologia  (Portugal), projects UID/MAT/04459/2013 and
PTDC/MAT-PUR/31089/2017 (Higher Structures and Applications).


\begin{thebibliography}{10}

\bibitem{baez-huerta}
{\sc Baez, J. C., and Huerta, J.}
\newblock An invitation to higher gauge theory.
\newblock {\em Gen. Relativity Gravitation 43}, no. 9 (2011), 2335--2392.


\bibitem{basch-hgt}
{\sc Baez, J., and Schreiber, U.}
\newblock Higher gauge theory: 2-connections on 2-bundles.
\newblock In {\em Categories in Algebra, Geometry and Mathematical Physics},
  {A. Davydov et al}, Ed., no.~431 in Contemp. Math. AMS, Providence, Rhode
  Island, 2007, pp.~7--30.

\bibitem{bartels}
{\sc Bartels, T.}
\newblock {\em Higher Gauge Theory I: 2-Bundles}.
\newblock PhD thesis, University of California, Riverside, 2004.
\newblock {\tt{arXiv:math/0410328}}.

\bibitem{braganca-picken-i}
{\sc  Bragança, D., and Picken, R.}
\newblock {\em Invariants and TQFT's for cut cellular surfaces from finite groups}.
\newblock  {\em Bol. Soc. Port. Mat. 74}, (2016) 17--44.


\bibitem{braganca-picken-ii}
{\sc  Bragança, D., and Picken, R.}
\newblock {\em Invariants and TQFT's for cut cellular surfaces from finite 2-groups}.
\newblock  {\em Bol. Soc. Port. Mat. } to appear.
\newblock {\tt{arXiv:math/1512.08263}}.


\bibitem{brown-higgins-sivera}
{\sc Brown, R., Higgins, P.J., Sivera, R.} 
\newblock {\em Nonabelian Algebraic Topology, Filtered Spaces, Crossed Complexes, Cubical
Homotopy Groupoids}.
\newblock EMS Tracts Math., vol. 15, European Mathematical Society, 2010.


\bibitem{brown-mackenzie}
{\sc Brown, R..; Mackenzie, K. C. H.}
\newblock Determination of a double Lie groupoid by its core diagram.
\newblock {\em J. Pure Appl. Algebra 80}, 3 (1992) 237--272.



\bibitem{bullivant-et-al}
{\sc Bullivant, A., Calcada, M., Kádár, Z.,  Faria Martins, J. and Martin, P. }
\newblock {\em Higher lattices, discrete two-dimensional holonomy and topological
phases in (3+1)D with higher gauge symmetry}.
\newblock {\tt{arXiv:1702.00868v4}}.

\bibitem{higgins-mackenzie}
{\sc Higgins, P. J.; Mackenzie, K.}
\newblock Algebraic constructions in the category of Lie algebroids.
\newblock {\em J. Algebra 129}, 1 (1990) 194--230.


\bibitem{martinspicken}
{\sc Faria Martins, J., and Picken, R.}
\newblock On two-dimensional holonomy.
\newblock {\em Trans. Amer. Math. Soc. 362}, 11 (2010) 5657--5695.


\bibitem{martinspickenii}
{\sc Faria Martins, J., and Picken, R.}
\newblock Surface holonomy for non-abelian 2-bundles via double groupoids.
\newblock {\em Adv. Math. 226}, 4 (2011), 3309--3366.

\bibitem{martinsporter}
{\sc Faria Martins, J., and Porter, T.}
\newblock On {Y}etter's invariant and an extension of the {D}ijkgraaf-{W}itten
  invariant to categorical groups.
\newblock {\em  Theory Appl. Categ. 18}, 4 (2007), 118--150.
\newblock {\tt{http://www.tac.mta.ca/tac/volumes/18/4/18-04.pdf}}.



\bibitem{morton-2lin} 
{\sc Morton, J.~C.}
\newblock Cohomological twisting of 2-linearization and extended TQFT.
\newblock {\em J. Homotopy Relat. Struct. 10}, 2 (2015), 127--187.
 

\bibitem{morton-dlbicat}
{\sc Morton, J.~C.}
\newblock Double bicategories and double cospans.
\newblock {\em J. Homotopy Relat. Struct. 4}, 1 (2009), 389--428.



\bibitem{morton-picken-i}
{\sc Morton, J.~C., and Picken, R.}
\newblock Transformation double categories associated to 2-group actions.
\newblock {\em Theory Appl. Categ. 30}  Paper No. 43, (2015), 1429--1468.

\bibitem{morton-picken-iii} {\sc Morton, J.~C., and Picken, R.}, in preparation

\bibitem{nelsonpicken}
{\sc Nelson, J.~E., and Picken, R.~F.}
\newblock Quantum geometry from 2 + 1 AdS quantum gravity on the torus.
\newblock {\em Gen. Relativity Gravitation}, 43 (2011), 777--795.


\bibitem{schreiberwaldorfii}
{\sc Schreiber, U., and Waldorf, K.}
\newblock Smooth functors vs. differential forms.
\newblock {\em Homology, Homotopy Appl. 13}, 1 (2011), 143--203.

\bibitem{soncinizucchini}
{\sc Soncini, E., and Zucchini, R.}
\newblock A new formulation of higher parallel transport in higher gauge theory.
\newblock {\em J. Geom. Phys. 95} (2015), 28--73.



\bibitem{yettqft}
{\sc Yetter, D.}
\newblock Topological quantum field theories associated to finite groups and
  crossed $G$-sets.
\newblock {\em J. Knot Theory Ramifications 1}, 1 (1992), 1--20.

\end{thebibliography}
\end{document}